\newif\ifIEEE
\newif\ifPAGELIMIT
\IEEEtrue
\PAGELIMITfalse
\ifPAGELIMIT
\IEEEtrue
\fi
\ifIEEE
    \ifPAGELIMIT
        \documentclass[conference,final,letterpaper]{IEEEtran}
    \else   
        \documentclass[journal,final,letterpaper]{IEEEtran}
    \fi
    \newcommand{\bibauthor}[1]{#1}
    \newcommand{\bibpaper}[1]{``#1''}
    \newcommand{\Footnotetext}[2]
    {
        \begin{figure}[!b]
        \footnotesize\vspace{-3ex}\hrulefill\hfill
        \makebox[0em]{}\hfill\makebox[0em]{}%
                                          \par${}^{#1}$ #2\vspace{-.6ex}
        \end{figure}
        \addtocounter{figure}{0}
     }
\else
    \documentclass[12pt]{article}
    \usepackage{amsthm}
    \newcommand{\bibauthor}[1]{\textsc{#1}}
    \newcommand{\bibpaper}[1]{\textsl{#1}}
    \newenvironment{IEEEkeywords}{\begin{small}%
                                  \textbf{Index Terms} ---}{\end{small}}
\fi
\usepackage{amsfonts,bm}
\usepackage{amsthm}
\usepackage{pict2e}
\usepackage{color}
\newcommand{\bibbook}[1]{\textit{#1}}

\newcommand{\bibperiodical}[1]{\textit{#1}}

\ifPAGELIMIT
\fi
\newtheorem{theorem}{\indent Theorem}

\newtheorem{lemma}[theorem]{\indent Lemma}
\newtheorem{corollary}[theorem]{\indent Corollary}

\theoremstyle{remark}
\newtheorem{remark}{\indent Remark}
\theoremstyle{definition}
\newtheorem{example}{\indent Example}
\renewcommand{\mathbf}[1]{{\bm{#1}}}     
\newcommand{\encoder}{{\mathcal{E}}}
\newcommand{\bldw}{{\mathbf{w}}}
\newcommand{\bldx}{{\mathbf{x}}}
\newcommand{\bldy}{{\mathbf{y}}}
\newcommand{\bldz}{{\mathbf{z}}}
\newcommand{\bldeta}{{\mathbf{\eta}}}
\newcommand{\bldomega}{{\mathbf{\omega}}}
\newcommand{\bldmu}{{\mathbf{\mu}}}
\newcommand{\bldxi}{{\mathbf{\xi}}}
\newcommand{\FF}{{\mathcal{F}}}
\newcommand{\XX}{{\mathcal{X}}}
\newcommand{\ZZ}{{\mathcal{Z}}}
\newcommand{\capacity}{{\mathsf{cap}}}
\newcommand{\parity}{{\mathsf{b}}}
\newcommand{\List}{{\mathcal{L}}}
\newcommand{\Labels}{{\mathcal{B}}}
\newcommand{\Prefix}{{\mathcal{P}}}
\newcommand{\Kraft}{{\mathsf{K}}}
\newcommand{\length}{{\mathsf{r}}}
\newlength{\figunit}
\ifPAGELIMIT
    \newcommand{\figfont}{\scriptsize}
    \newcommand{\ifandonlyif}{iff}
    \newcommand{\withoutlossofgenerality}{w.l.o.g.}
    \newcommand{\withrespectto}{w.r.t.}
    \newcommand{\Withrespectto}{W.r.t.}
    \newcommand{\respectively}{resp.}
    \newcommand{\Section}{Section}
    \newcommand{\Sections}{Sections}
    \newcommand{\Figure}{Figure}
\else
\newcommand{\figfont}{\normalsize}
    \newcommand{\ifandonlyif}{if and only if}
    \newcommand{\withoutlossofgenerality}{without loss of generality}
    \newcommand{\withrespectto}{with respect to}
    \newcommand{\Withrespectto}{With respect to}
    \newcommand{\respectively}{respectively}
    \newcommand{\Section}{Section}
    \newcommand{\Sections}{Sections}
    \newcommand{\Figure}{Figure}
\fi

\newcommand{\Title}{Variable-Length Constrained Coding
and Kraft Conditions: The Parity-Preserving Case}
\newcommand{\Namea}{Ron M. Roth}
\newcommand{\Nameb}{Paul H. Siegel}
\newcommand{\Addressa}{Computer Science Department}
\newcommand{\Addressatwo}{Technion, Haifa 320003, Israel}
\newcommand{\Addressb}{ECE Department and CMRR}
\newcommand{\Addressbtwo}{UC San Diego, La Jolla, CA 92023, USA}

\newcommand{\Emaila}{ronny@cs.technion.ac.il}
\newcommand{\Emailb}{psiegel@ucsd.edu}
\newcommand{\Grant}{This work was supported by 
                    Grants~2015816 and~2018048 from
                    the United-States--Israel Binational
                    Science Foundation (BSF),
                    by NSF Grant CCF-BSF-1619053,
                    and by Grants~1396/16 and~1713/20 from
                    the Israel Science Foundation.
                    \ifPAGELIMIT
                    \else
                    This work, under the title
                    ``On parity-preserving
                    variable-length constrained coding,''
                    was presented in part at
                    the IEEE Int'l Symposium on Information
                    Theory (ISIT), June 2020.
                    \fi
}
\newcommand{\Addressalt}{This work was done in part while R.M. Roth
            was visiting
\ifPAGELIMIT
            CMRR,
\else
            the Center for Memory and Recording Research (CMRR),
\fi
            UC San Diego}
\newcommand{\Thnxa}{\Namea\ is with the \Addressa, \Addressatwo.
                    \Addressalt.
                    Email: \Emaila}
\newcommand{\Thnxb}{\Nameb\ is with the \Addressb, \Addressbtwo.
                    Email: \Emailb}

\begin{document}
\ifIEEE
    \title{\Title}
       \ifPAGELIMIT
           \author{\IEEEauthorblockN{\Namea\vspace{-1ex}}\\
                   \IEEEauthorblockA{\Addressa\\ 
                                     \Addressatwo\\
                                     \Emaila\vspace{-2ex}}
                   \and
                   \IEEEauthorblockN{\Nameb\vspace{-1ex}}\\
                   \IEEEauthorblockA{\Addressb\\ 
                                     \Addressbtwo\\
                                     \Emailb\vspace{-2ex}}
           }
       \else
           \author{\Namea
                   \quad\quad
                   \Nameb
                   \thanks{\Grant}
                   \thanks{\Thnxa}
                   \thanks{\Thnxb}}
       \fi
\else
    \title{\textbf{\Title}\thanks{\Grant}}
    \author{\textsc{\Namea}\thanks{\Thnxa}
    \and
           \textsc{\Nameb}\thanks{\Thnxb}
    }
\fi
\maketitle


\begin{abstract}
Previous work by the authors on parity-preserving fixed-length
constrained encoders is extended to the variable-length case.
Parity-preserving variable-length encoders are formally defined,
\ifPAGELIMIT
    and
\else
and, to this end, Kraft conditions are developed for
the parity-preserving variable-length setting. Then,
\fi
a necessary and sufficient condition is presented for
the existence of
deterministic parity-preserving variable-length encoders for
a given constraint. Examples are provided that show
that there are coding ratios
where parity-preserving variable-length encoders exist,
while fixed-length encoders do not.
\end{abstract}

\ifPAGELIMIT
\else
\begin{IEEEkeywords}
Constrained codes,
Kraft inequality,
Parity-preserving encoders,
Variable-length encoders.
\end{IEEEkeywords}
\fi

\ifPAGELIMIT
    \Footnotetext{\quad}{\Addressalt. \Grant}
\fi

\section{Introduction}
\label{sec:introduction}

In mass storage platforms, such as magnetic and optical disks,
user data is mapped (encoded) to binary sequences that satisfy certain
combinatorial constraints. One common example of such a constraint
is the $(d,k)$-runlength-limited (RLL) constraint, where
the runs of $0$'s in a sequence are limited to have lengths
at least~$d$ (to avoid inter-symbol interference)
and at most~$k$ (to allow clock resynchronization)~\cite{Immink}.
In virtually all applications, the encoder takes the form
of a finite state machine, where user data is broken into binary
blocks, and each block is mapped, in a state-dependent manner,
into a binary \emph{codeword}, so that the concatenation of
the generated codewords satisfies the RLL constraint. 
In the case of \emph{fixed-length encoders},
the input blocks all have the same length~$p$,
and the codewords all have the same length~$q$,
for prescribed positive integers~$p$ and~$q$.
The coding rate is then $p:q$.

In the mentioned storage applications, there is also a need
to control the DC content of the recorded modulated sequence.
One commonly used strategy to achieve DC control is allowing
input blocks to be mapped to more than one codeword,
and the encoder then selects the codeword that yields
a better DC suppression~\cite[p.~29]{MRS}. 
In the Blu-ray standard, this strategy is applied 
through the use of \emph{parity-preserving encoders}:
such encoders map each input block to a codeword that has
the same parity (of the number of $1$s), and DC control is achieved by
reserving one bit in the input block to be set to a value
that minimizes the DC contents~\cite[\S 11.4.3]{Immink},
\cite{KI},
\ifPAGELIMIT
    \cite{MSIWUYN}--\cite{NY},
\else
\cite{MSIWUYN},
\cite{NKHSEKDW},
\cite{NY},
\fi
\cite{WIXC}.

Most constructions of parity-preserving encoders that were proposed
for commercial use were obtained by ad-hoc methods. 
In~\cite{RS1}, we initiated a study 
of bi-modal encoders
(which include parity-preserving encoders as a special case),
focusing on fixed-length encoders; we will summarize
the concepts that pertain to the fixed-length case,
along with the main results of~\cite{RS1},
as part of the background that we provide
in \Section~\ref{sec:fixedlength} below.
On the other hand, the existing ad-hoc
parity-preserving constructions typically have \emph{variable length},
where the length~$p$ of the input block and the length~$q$ 
of the respective codeword may depend on the encoder state,
as well as on the input sequence (the \emph{coding ratio}, $p/q$,
nevertheless, is still fixed).

In this work, we present several results on
parity-preserving variable-length encoders
(in short, parity-preserving VLEs), focusing on deterministic encoders.
To put our results into
perspective, we mention that even in the ordinary setting 
(where parity preservation is not required), the known tools for
analyzing and synthesizing VLEs are much less
developed, compared to the fixed-length case.
A summary of relevant (and mostly known)
results on (ordinary) VLEs is provided in \Section~\ref{sec:vlg-vle}.
In
\ifPAGELIMIT
    \Section~\ref{sec:PPVLE}
\else
\Sections~\ref{sec:Kraft}--\ref{sec:PPVLE}
\fi
we turn to the parity-preserving setting.
Much of the discussion in
\ifPAGELIMIT
    that section
\else
    those sections
\fi
deals in fact with
the definition of parity-preserving VLEs,
as it entails a (nontrivial) extension of the known Kraft conditions
on variable-length coding to the parity-preserving case.
\ifPAGELIMIT
    Our main result
\else
This extension, which may be of independent interest,
is developed in \Section~\ref{sec:Kraft},
followed in \Section~\ref{sec:PPVLE} by our main result, which
\fi
is a necessary and sufficient condition for
the existence of parity-preserving VLEs that are deterministic.
We present several examples that demonstrate
the advantages that parity-preserving VLEs may have over
their fixed-length counterparts, in terms of the attainable
coding ratios and encoding--decoding complexity.

\section{Fixed-length graphs and encoders}
\label{sec:fixedlength}

In this section, we extract from~\cite[Chapters~2--5]{MRS}
several basic definitions and properties pertaining to ordinary
(namely, fixed-length) graphs and fixed-length encoders.
We then quote the main result of~\cite{RS1},
which applies, in particular,
to parity-preserving fixed-length encoders.

\ifPAGELIMIT
\else
\subsection{Graphs and constraints}
\label{sec:fixedlengthgraphs}
\fi

A (finite labeled directed ordinary) graph is a graph $G = (V,E,L)$
where~$V$ is a nonempty finite set of states,
$E$ is a finite set of edges, and 
$L : E \rightarrow \Sigma$ is an edge labeling.
\ifPAGELIMIT
\else
We say that a (finite) word~$\bldw$ over~$\Sigma$ is
\emph{generated} by a path~$\pi$ in~$G$
if~$\bldw$ is obtained by reading the labels along~$\pi$;
the length of~$\bldw$ then equals the length of~$\pi$
(being the number of edges along~$\pi$).
\fi
A graph~$G$ is \emph{deterministic} if no two outgoing edges from
the same state in~$G$ have the same
\ifPAGELIMIT
    labels, and it is \emph{lossless} if
\else
label.
A deterministic graph is a special case of a graph with
\emph{finite anticipation}: the anticipation of a graph~$G$
is the smallest integer $a \ge 0$
(if any) such that any two paths with the same initial state
that generate the same word of length~$a{+}1$ must have the same
initial edge (a deterministic graph corresponds to the case
where the anticipation is~$0$). Having finite anticipation, in turn,
implies (generally) that the graph is \emph{lossless}:
\fi
no two paths with the same initial state and
the same terminal state generate the same word. 

A \emph{constraint} $S$ over an alphabet~$\Sigma$ is
the set of all words that are generated by
paths in a graph~$G$; we then say that~$G$ \emph{presents} $S$
and write $S = S(G)$. Every constraint~$S$ can be presented
by a deterministic graph.
The \emph{capacity} of~$S$ is defined by
$\capacity(S) =
\lim_{\ell \rightarrow \infty} (1/\ell) \log_2 |S \cap \Sigma^\ell|
\ifPAGELIMIT
$.
\else
$ (where, by sub-additivity, the limit indeed exists).
\fi
It is known that
$\capacity(S) = \log_2 \lambda(A_G)$ where~$\lambda(A_G)$ denotes
the spectral radius (Perron eigenvalue)
of the adjacency matrix~$A_G$ of any lossless
(in particular, deterministic) presentation~$G$ of~$S$.

A graph~$G$ is \emph{irreducible} if it is
strongly connected,
namely, for any two states~$u$ and~$v$ in~$G$
there is a path from~$u$ to~$v$.
A constraint~$S$ is irreducible
if it can be presented by a deterministic irreducible graph.
For irreducible constraints,
there is a unique deterministic graph presentation
that has the smallest number of states; such a presentation is called
the \emph{Shannon cover} of~$S$.

\begin{example}
\label{ex:twostates}
Let~$S$ be the constraint over the alphabet
$\Sigma = \{ a, b, c, d \}$ which is presented by
the graph~$G$ in \Figure~\ref{fig:twostates}.
The graph~$G$
\ifPAGELIMIT
    is the Shannon cover of~$S$, and
    $\capacity(S) = \log_2 \lambda(A_G) = \log_2 2 = 1$.\qed
\else
is deterministic and irreducible
(in fact, it is the Shannon cover of~$S$).
The adjacency matrix of~$A_G$ is given by
\[
A_G =
\left(
\begin{array}{cc}
1 & 2 \\
1 & 0
\end{array}
\right) \; ,
\]
and $\lambda(A_G) = 2$, with a respective eigenvector
$\bldx = (2 \;\, 1)^\top$.
Hence, $\capacity(S) = \log_2 \lambda(A_G) = \log_2 2 = 1$.\qed
\fi
\begin{figure}[hbt]
\begin{center}
\thicklines
\setlength{\unitlength}{\figunit}
\figfont
\ifPAGELIMIT
    \begin{picture}(105,25)(-35,-05)
\else
\begin{picture}(105,45)(-35,-15)
\fi
    \multiput(000,000)(060,000){2}{\circle{20}}
    \qbezier(051,4.5)(030,14.5)(009,4.5)
    \put(051,4.5){\vector(2,-1){0}}
    \qbezier(054,8.3)(030,28.8)(006,8.3)
    \put(054,8.3){\vector(5,-4){0}}
    \qbezier(051,-4.5)(030,-14.5)(009,-4.5)
    \put(009,-4.5){\vector(-2,1){0}}

    \qbezier(-9.26,004)(-30,015)(-30,000)
    \qbezier(-9.26,-04)(-30,-15)(-30,000)
    \put(-9.26,004){\vector(2,-1){0}}

    \put(000,000){\makebox(0,0){$\alpha$}}
    \put(060,-.5){\makebox(0,0){$\beta$}}

    \put(-32,000){\makebox(0,0)[r]{$a$}}
    \put(030,020.8){\makebox(0,0)[b]{$b$}}
    \put(030,011.0){\makebox(0,0)[b]{$c$}}
    \put(030,-11.5){\makebox(0,0)[t]{$d$}}
\end{picture}
\thinlines
\setlength{\unitlength}{1pt}
\end{center}
\caption{Graph~$G$ for Example~\protect\ref{ex:twostates}.}
\label{fig:twostates}
\end{figure}
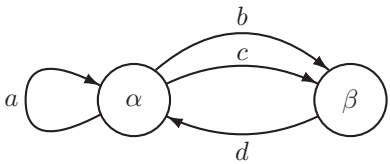
\end{example}

The power~$G^t$ of a graph $G = (V,E,L)$ is
the graph with the same set of states~$V$
and edges that are the paths of length~$t$ in~$G$; the label
of an edge in~$G^t$ is the length-$t$ word generated by the path.
For $S = S(G)$ the power~$S^t$ is defined as~$S(G^t)$.

\ifPAGELIMIT
\else
\subsection{Fixed-length encoders}
\label{sec:fixedlengthencoders}
\fi

Given a constraint~$S$ and a positive
integer~$n$,
a \emph{(fixed-length) $(S,n)$-encoder} is a lossless graph~$\encoder$
such that $S(\encoder) \subseteq S$ and
each state has out-degree~$n$.
An $(S,n)$-encoder exists
\ifandonlyif\ $\log_2 n \le \capacity(S)$.
In a \emph{tagged} $(S,n)$-encoder,
each edge is assigned an input tag from
a finite alphabet~$\Upsilon$ of size~$n$,
such that edges outgoing from the same state have distinct tags.
\ifPAGELIMIT
\else
A tagged encoder is
\emph{$(m,a)$-sliding-block decodable} if all paths that generate
a given word of length $m{+}a{+}1$ share the same tag
on their $(m{+}1)$st edges.
\fi

A \emph{(tagged) rate $p:q$ encoder} for a constraint~$S$ is
a tagged $(S^q,2^p)$-encoder (the tag alphabet~$\Upsilon$
is then assumed to be $\{ 0, 1 \}^p$); such an encoder exists
\ifandonlyif\ $p/q \le \capacity(S)$.

\ifPAGELIMIT
\else
Given a square nonnegative integer matrix~$A$ and a positive
integer~$n$, an \emph{$(A,n)$-approximate eigenvector}
is a nonnegative nonzero integer vector~$\bldx$ that satisfies
the inequality $A \bldx \ge n \bldx$ componentwise.
The set of all $(A,n)$-approximate eigenvectors will be denoted by
$\XX(A,n)$.
Given a constraint~$S$ presented by a deterministic graph~$G$ and
a positive integer~$n$, the state-splitting algorithm provides
a method for transforming $G$, through
an $(A_G,n)$-approximate eigenvector, into
an $(S,n)$-encoder with finite anticipation.

\begin{example}
\label{ex:fixedlength}
Letting~$G$ and~$S$ be as in Example~\ref{ex:twostates},
the graph in \Figure~\ref{fig:fixedlength} is 
a tagged $(S,2)$-encoder (or a rate $1:1$ encoder for~$S$),
where each edge is assigned a tag from $\{ 0, 1 \}$
(the notation ``$s/w$'' next to an edge 
specifies the tag~$s$ and the label~$w$ of the edge).
The encoder is obtained by splitting state~$\alpha$
in~$G$ into two states: state~$\alpha'$ inherits
the outgoing edges labeled by~$b$ and~$c$,
and state~$\alpha''$ inherits the self-loop labeled~$a$
(this splitting is implied by
the $(A_G,2)$-approximate eigenvector $\bldx = (2 \;\, 1)^\top$,
which is also a true eigenvector of~$A_G$,
where state~$\alpha$ in~$G$ is
assigned a weight of~$2$, and state~$\beta$ has weight~$1$).
The encoder is not deterministic, but it is
$(0,1)$-sliding-block decodable (and hence has anticipation~$1$):
a label of an edge uniquely determines the initial state of the edge
and, so, any word $\bldw \in S$ of length~$2$ uniquely determines
the first edge of any path that generates~$\bldw$.\qed
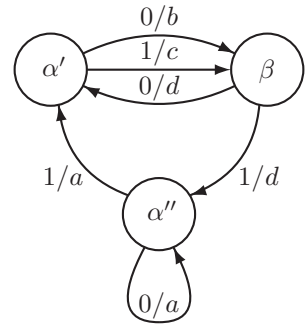
\begin{figure}[hbt]
\begin{center}
\thicklines
\setlength{\unitlength}{\figunit}
\figfont
\begin{picture}(080,085)(-10,-70)
    \multiput(000,000)(060,000){2}{\circle{20}}
    \put(030,-40){\circle{20}}
    \qbezier(051,4.5)(030,14.5)(009,4.5)
    \put(051,4.5){\vector(2,-1){0}}
    \qbezier(051,-4.5)(030,-14.5)(009,-4.5)
    \put(009,-4.5){\vector(-2,1){0}}

    \put(010,000){\vector(1,0){40}}

    \qbezier(38.6,-34.8)(56.2,-28.2)(57.8,-9.6)
    \put(38.6,-34.8){\vector(-3,-1){0}}

    \qbezier(21.4,-34.8)(3.8,-28.2)(2.2,-9.6)
    \put(02.2,-9.6){\vector(-1,6){0}}

    \qbezier(034,-49.26)(045,-70)(030,-70)
    \qbezier(026,-49.26)(015,-70)(030,-70)
    \put(034,-49.26){\vector(-1,2){0}}

    \put(001,001){\makebox(0,0){$\alpha'$}}
    \put(060,-.5){\makebox(0,0){$\beta$}}
    \put(031,-39){\makebox(0,0){$\alpha''$}}

    \put(030,010.5){\makebox(0,0)[b]{$0/b$}}
    \put(030,001.0){\makebox(0,0)[b]{$1/c$}}
    \put(030,-08.5){\makebox(0,0)[b]{$0/d$}}
    \put(052,-30){\makebox(0,0)[l]{$1/d$}}
    \put(009,-30){\makebox(0,0)[r]{$1/a$}}
    \put(030,-69){\makebox(0,0)[b]{$0/a$}}
\end{picture}
\thinlines
\setlength{\unitlength}{1pt}
\end{center}
\caption{Tagged fixed-length $(S,2)$-encoder
for Example~\protect\ref{fig:fixedlength}.}
\label{fig:fixedlength}
\end{figure}
\end{example}

\subsection{Parity-preserving fixed-length encoders}
\label{sec:fixedlengthparitypreservingencoders}
\fi

Let~$\Sigma$ be an alphabet
and fix a partition $\{ \Sigma_0, \Sigma_1 \}$ of $\Sigma$.
The symbols in~$\Sigma_0$ (\respectively, $\Sigma_1$)
will be referred to as
the \emph{even} (\respectively, \emph{odd}) symbols of $\Sigma$.
Extending the definition of parity to words,
we say that a word $\bldw$ over~$\Sigma$ is even (\respectively, odd)
if~$\bldw$ contains an even (\respectively, odd)
number of symbols from~$\Sigma_1$.
The set of even (\respectively, odd) words in~$\Sigma^t$ will be 
denoted by~$(\Sigma^t)_0$ (\respectively, $(\Sigma^t)_1$).
In the practical scenario where $\Sigma = \{ 0, 1 \}^p$,
with $\Sigma_0$ and $\Sigma_1$ consisting of
the binary $p$-tuples with even and odd parity, \respectively\
(according to the common meaning of parity),
a parity of a word in~$\Sigma^t$, too,
coincides with the ordinary meaning of this term.

Given a graph~$H$ with labeling in~$\Sigma$,
for $\parity \in \{ 0, 1 \}$,
we denote by~$H_\parity$ the subgraph of~$H$ containing only the edges
with labels in~$\Sigma_\parity$.

\ifPAGELIMIT
\else
\begin{example}
\label{ex:partitioning}
Let $\Sigma = \{ a, b, c, d \}$ and assume the partition
$\{ \Sigma_0, \Sigma_1 \}$, where
\begin{equation}
\label{eq:partitioning}
\Sigma_0 = \{ a, b \} \quad \textrm{and} \quad \Sigma_1 = \{ c, d \}
\; .
\end{equation}
For the graph~$G$ in \Figure~\ref{fig:twostates},
the subgraphs~$G_0$ and~$G_1$ \withrespectto\ this partition
are shown in Figures~\ref{fig:partitioning1}
and~\ref{fig:partitioning2}.\qed
\begin{figure}[hbt]
\begin{center}
\thicklines
\setlength{\unitlength}{\figunit}
\figfont
\begin{picture}(130,25)(-35,-10)
    \multiput(000,000)(060,000){2}{\circle{20}}
    \qbezier(051,4.5)(030,14.5)(009,4.5)
    \put(051,4.5){\vector(2,-1){0}}

    \qbezier(-9.26,004)(-30,015)(-30,000)
    \qbezier(-9.26,-04)(-30,-15)(-30,000)
    \put(-9.26,004){\vector(2,-1){0}}

    \put(000,000){\makebox(0,0){$\alpha$}}
    \put(060,-.5){\makebox(0,0){$\beta$}}
    
    \put(-32,000){\makebox(0,0)[r]{$a$}}
    \put(030,012.0){\makebox(0,0)[b]{$b$}}
\end{picture}
\thinlines
\setlength{\unitlength}{1pt}
\end{center}
\caption{Subgraph~$G_0$ for Example~\protect\ref{ex:partitioning}.}
\label{fig:partitioning1}
\end{figure}
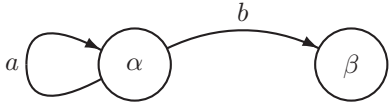
\begin{figure}[hbt]
\begin{center}
\thicklines
\setlength{\unitlength}{\figunit}
\figfont
\begin{picture}(130,25)(-35,-10)
    \multiput(000,000)(060,000){2}{\circle{20}}
    \qbezier(051,4.5)(030,14.5)(009,4.5)
    \put(051,4.5){\vector(2,-1){0}}
    \qbezier(051,-4.5)(030,-14.5)(009,-4.5)
    \put(009,-4.5){\vector(-2,1){0}}

    \put(000,000){\makebox(0,0){$\alpha$}}
    \put(060,-.5){\makebox(0,0){$\beta$}}
    
    \put(030,012.0){\makebox(0,0)[b]{$c$}}
    \put(030,-6.5){\makebox(0,0)[b]{$d$}}
\end{picture}
\thinlines
\setlength{\unitlength}{1pt}
\end{center}
\caption{Subgraph~$G_1$ for Example~\protect\ref{ex:partitioning}.}
\label{fig:partitioning2}
\end{figure}
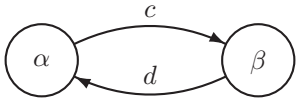
\end{example}
\fi

Let $S$ be a constraint over an alphabet~$\Sigma$, fix
a partition $\{ \Sigma_0, \Sigma_1 \}$ of $\Sigma$,
and let~$n_0$ and~$n_1$ be positive
integers. A \emph{(fixed-length) $(S,n_0,n_1)$-encoder} $\encoder$ is
an $(S,n_0{+}n_1)$-encoder such that for each~$\parity \in \{ 0, 1 \}$,
the subgraph~$\encoder_\parity$ is an~$(S,n_\parity)$-encoder.
A \emph{rate~$p:q$ parity-preserving (fixed-length) encoder} for~$S$ is
a tagged $(S^q,2^{p-1},2^{p-1})$-encoder in which
the tag (in $\{ 0, 1 \}^p$) that is assigned to each edge has
the same parity as the edge label (when seen as a word in $\Sigma^q$).
\ifPAGELIMIT
\else
Conversely, in any $(S^q,2^{p-1},2^{p-1})$-encoder we can assign
tags from $\{ 0, 1 \}^p$ to the edges so that the parities of the tags
and the labels match on each edge.

\begin{example}
\label{ex:main}
Letting~$\Sigma$ and~$S$ be as in Example~\ref{ex:twostates},
the $(S,2)$-encoder in \Figure~\ref{fig:fixedlength}
is \emph{not} an $(S,1,1)$-encoder 
\withrespectto\ the partition~(\ref{eq:partitioning}) of~$\Sigma$,
since both outgoing edges from state~$\alpha'$
(\respectively, state~$\alpha''$) have the same parity.
In fact, using Theorem~\ref{thm:main}(a) below,
it was shown in~\cite{RS1} that for the constraint~$S$
and for the partition~(\ref{eq:partitioning}),
there is no $(S^t,2^{t-1},2^{t-1})$-encoder
for any positive integer~$t$, namely,
a coding ratio of~$1$ cannot be achieved 
by any parity-preserving (fixed-length) encoder, for any~$t$.\qed
\end{example}
\fi

The next theorem follows from the results of~\cite{RS1}
(see Theorem~1, Corollary~5, and \S III-A therein).
\ifPAGELIMIT
    For a square nonnegative integer matrix~$A$ and
    a positive integer~$n$, denote by $\XX(A,n)$ the set of all
    nonnegative nonzero integer vectors~$\bldx$ that satisfy
    the inequality $A \bldx \ge n \bldx$ componentwise.
\fi

\begin{theorem}[\cite{RS1}]
\label{thm:main}
Let $S$ be an irreducible constraint, presented by
an irreducible deterministic graph $G$, and let~$n_0$ and~$n_1$ be
positive integers. Then the following holds.
\begin{list}{}{\settowidth{\labelwidth}{\textit{(a)}}%
               \settowidth{\leftmargin}{\textit{(a)}.....}}
\item[(a)]
There exists an $(S,n_0,n_1)$-encoder, if and only if
$\XX(A_{G_0},n_0) \cap \XX(A_{G_1},n_1) \ne \emptyset$.
\item[(b)]
There exists a deterministic $(S,n_0,n_1)$-encoder, if and only if
$\XX(A_{G_0},n_0) \cap \XX(A_{G_1},n_1)$
contains a $0\mbox{--}1$ vector.
\end{list}
\end{theorem}

\section{Variable-length graphs and encoders}
\label{sec:vlg-vle}

In this section, we summarize several definitions and properties
relating to variable-length graphs 
and variable-length encoders (see also~\cite[\S 6.4]{MRS}).

\subsection{Variable-length graphs}
\label{sec:vlg}

In a variable-length graph (in short, VLG), the labels of the edges
may be words of any positive (finite) length over
the label alphabet~$\Sigma$; the length of the edge
is then defined as the length of its label.
Given a VLG~$H$, the constraint $S(H)$ that is presented by~$H$ is
defined as the set of all (consecutive) \emph{sub-words} of words
obtained by concatenating the labels that are read
along finite paths in~$H$.
Equivalently, $S(H)$
is the constraint presented by the (ordinary) graph~$G$
obtained from~$H$ by replacing each length-$\ell$ edge~$e$ in~$H$
by a path of~$\ell$ length-$1$ edges
(connected through newly introduced dummy states)
which generates the label of~$e$.
The \emph{follower set} of a state~$u$
\ifPAGELIMIT
    in~$H$
\else
in~$H$, denoted $\FF_H(u)$,
\fi
is the set of all \emph{prefixes} of words
that are generated by finite paths that start at~$u$.

A VLG~$H$ is called deterministic if
the labels of the outgoing edges from each state in~$H$ form
a \emph{prefix-free} list, namely, no label is a prefix
of any other label. The notions of losslessness and irreduciblity
carry over from ordinary
\ifPAGELIMIT
    graphs.
\else
graphs:
$H$ is lossless if no two paths in~$H$ that
share the same initial state and terminal state
generate the same word,
and it is irreducible if it is strongly connected.
\fi

\begin{example}
\label{ex:vlg}
Letting~$G$ and~$S$ be as in Example~\ref{ex:twostates},
the VLG~$H$ in \Figure~\ref{fig:vle}
is irreducible and deterministic, and it presents~$S$, i.e.,
\ifPAGELIMIT
    $S(H) = S(G) = S$.\qed
\else
$S(H) = S(G) = S$.
In particular, we have $\FF_H(\alpha) = \FF_G(\alpha)$.\qed
\fi
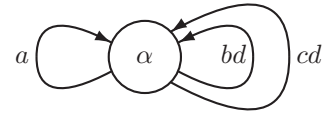
\begin{figure}[hbt]
\begin{center}
\thicklines
\setlength{\unitlength}{\figunit}
\figfont
\ifPAGELIMIT
    \begin{picture}(080,015)(-35,-05)
\else
\begin{picture}(080,035)(-35,-15)
\fi
    \put(000,000){\circle{20}}

    \qbezier(-9.26,004)(-30,015)(-30,000)
    \qbezier(-9.26,-04)(-30,-15)(-30,000)
    \put(-9.26,004){\vector(2,-1){0}}

    \qbezier(9.26,004)(030,015)(030,000)
    \qbezier(9.26,-04)(030,-15)(030,000)
    \put(9.26,004){\vector(-2,-1){0}}
    \qbezier(7.14,007)(040,025)(040,000)
    \qbezier(7.14,-07)(040,-25)(040,000)
    \put(7.14,007){\vector(-3,-2){0}}

    \put(000,000){\makebox(0,0){$\alpha$}}

    \put(-32,000){\makebox(0,0)[r]{$a$}}
    \put(028,001){\makebox(0,0)[r]{$bd$}}
    \put(042,001){\makebox(0,0)[l]{$cd$}}
\end{picture}
\thinlines
\setlength{\unitlength}{1pt}
\end{center}
\caption{VLG~$H$ for Example~\protect\ref{ex:vlg}.}
\label{fig:vle}
\end{figure}
\end{example}

\ifPAGELIMIT
\else
\begin{remark}
\label{rem:vlg}
The follower-set equality, $\FF_H(\alpha) = \FF_G(\alpha)$,
in Example~\ref{ex:vlg} is in fact an instance of
a more general property.
Let~$S$ be an irreducible constraint and let~$G$ be its Shannon cover
(i.e., the unique deterministic presentation of~$S$
with the smallest number of states).
Also, let~$H$ be an irreducible deterministic VLG that presents~$S$.
We can obtain from~$H$ an ordinary irreducible deterministic
graph~$H'$ (with length-$1$ edges) by transforming the outgoing
edges from each state in~$H$ into a tree. From the uniqueness of
the Shannon cover (and, specifically, from~\cite[Theorem~2.12(b)]{MRS})
we get that the follower sets of the states of~$H'$ coincide with
the follower sets of the states of~$G$.
Hence, for every state~$u$ in~$H$ there exists
a state~$v$ in~$G$ such that $\FF_H(u) = \FF_G(v)$.\qed
\end{remark}
\fi

\subsection{Variable-length encoders}
\label{sec:vle}

Let~$\Upsilon$ be a finite
alphabet\footnote{%
We use here the notation~$\Upsilon$ for an alphabet
(instead of~$\Sigma$) since in the context of variable-length encoders,
that alphabet will be the alphabet of tags.}
and let~$\List$ be a finite list of nonempty finite words
\ifPAGELIMIT
    over~$\Upsilon$.
\else
over~$\Upsilon$ (the empty word is the unique word of length~$0$).
\fi
We say that~$\List$ is \emph{exhaustive} if every word
over~$\Upsilon$ either has a prefix in~$\List$
or is a prefix of some word in~$\List$.
The next result is well known~\cite[p.~298]{Blahut}.

\begin{theorem}
\label{thm:Kraft}
Given an alphabet~$\Upsilon$
and a nonnegative integer sequence $\bldmu = (\mu_\ell)_{\ell \ge 1}$
with finite support,
there exists an exhaustive prefix-free list~$\List$ over~$\Upsilon$
such that
\[
\mu_\ell = \left| \List \cap \Upsilon^\ell \right|
\; , \quad
\ell = 1, 2, 3, \cdots \; ,
\]
\ifandonlyif~$\bldmu$
satisfies the Kraft inequality with equality, namely:
\begin{equation}
\label{eq:Kraft}
\sum_{\ell \ge 1}
\frac{\mu_\ell}{|\Upsilon|^\ell} = 1
\; .
\end{equation}
\end{theorem}

Let~$S$ be a constraint over an alphabet~$\Sigma$
and let~$n$ be a positive integer.
Also, let $\encoder = (V,E,L)$ be a VLG,
and for every $u \in V$ and $\ell \ge 1$,
denote by $\mu_\ell(u)$ the number of edges
of length~$\ell$ outgoing from~$u$ in~$\encoder$.
We say that~$\encoder$ is a \emph{variable-length $(S,n)$-encoder}
(in short, an \emph{$(S,n)$-VLE}) if the following conditions hold.
\begin{list}{}{\settowidth{\labelwidth}{\textup{(E2)}}%
               \settowidth{\leftmargin}{\textup{(E2...)}}}
\item[(E1)]
$\encoder$ is lossless,
\item[(E2)]
$S(\encoder) \subseteq S$, and---
\item[(E3)]
for every $u \in V$:
\[
\sum_{\ell \ge 1} \frac{\mu_\ell(u)}{n^\ell} = 1
\; .
\]
\end{list}
(This definition reduces to that of a fixed-length $(S,n)$-encoder
when $\mu_\ell(u) = 0$ for every $u \in V$ and $\ell > 1$.)

Extending now the notion of tagging to the variable-length case,
let~$\Upsilon$ be a (base tag) alphabet of size $|\Upsilon| = n$.
A tagging of an $(S,n)$-VLE $\encoder$ is an assignment of
input tags---namely, words over~$\Upsilon$---to the edges of
$\encoder$, such that:
\begin{list}{}{\settowidth{\labelwidth}{\textup{(T2)}}%
               \settowidth{\leftmargin}{\textup{(T2...)}}}
\item[(T1)]
the length of each input tag equals the length of
(the label of) the edge, and---
\item[(T2)]
the input tags of the outgoing edges from each state in~$\encoder$
form an exhaustive prefix-free list over $\Upsilon$.
\end{list}

Theorem~\ref{thm:Kraft} and condition~(E3) guarantee
that every $(S,n)$-VLE can be tagged
consistently with conditions~(T1)--(T2).
Condition~(T1) means that the \emph{coding ratio}
is fixed to be~$1$ at all edges, regardless of their length
\ifPAGELIMIT
    (by grouping symbols and tags into nonoverlapping blocks,
\else
(as we argue in Remark~\ref{rem:codingratio} below,
\fi
any fixed coding ratio can be reduced to
the case of a coding ratio of~$1$).
We note that this is the variable-length encoding model assumed
in~\cite{AFKM}, \cite{Beal1},
\ifPAGELIMIT
   \cite{Franaszek2}.
\else
\cite{Franaszek2},
and this model is more restrictive than the one
in~\cite{HMS}, where the coding ratio needs to be constant only along
\emph{cycles} in the encoder (see \Figure~\ref{fig:vleHMS} below).
\fi

\begin{example}
\label{ex:vle}
Letting~$\Sigma$ and~$S$ be as in Example~\ref{ex:twostates},
the graph~$H$ in \Figure~\ref{fig:vle} is a deterministic $(S,2)$-VLE.
Taking $\Upsilon = \{ 0, 1 \}$, one possible tag assignment to
(the labels of) the edges of~$H$ is
\ifPAGELIMIT
   given by
   $0  \leftrightarrow a$,
   $10 \leftrightarrow bd$, and
   $11 \leftrightarrow cd$.
   The
\else
shown in Table~\ref{tab:vle}.
\begin{table}[hbt]
\caption{Possible tag assignment for
the encoder in \Figure~\ref{fig:vle}.}
\label{tab:vle}
\[
\begin{array}{lcl}
0  & \leftrightarrow & a  \\
10 & \leftrightarrow & bd \\
11 & \leftrightarrow & cd \\
\end{array}
\]
\end{table}
The coding rate is $1:1$ when the input tag is~$0$,
and $2:2$ when the input tag starts with a~$1$;
namely, the
\fi
coding ratio at each state is~$1$, 
so this encoder is capacity-achieving.
Note that this tag assignment is parity-preserving
\ifPAGELIMIT
    \withrespectto\ the following partition of~$\Sigma$:
    \begin{equation}
    \label{eq:partitioning}
    \Sigma_0 = \{ a, b \} \quad \textrm{and} \quad \Sigma_1 = \{ c, d \}
    \; .
    \end{equation}
\else
\withrespectto\ the partition~(\ref{eq:partitioning}) of~$\Sigma$.
\fi
In contrast,
\ifPAGELIMIT
    using Theorem~\ref{thm:main}(a), it was shown in~\cite{RS1}
\else
recall from Example~\ref{ex:main}
\fi
that for this partition,
a coding rate of $t:t$ cannot be achieved 
by any parity-preserving \emph{fixed-length} encoder for~$S$
for any positive integer~$t$.\qed
\end{example}

\begin{example}
\label{ex:vlealt}
Letting~$\Sigma$ and~$S$ be as in Example~\ref{ex:twostates},
the graph~$\encoder$ in \Figure~\ref{fig:vlealt} presents
another $(S,2)$-VLE.
\ifPAGELIMIT
    The coding ratio at each state is~$1$, 
\else
The coding rate at state~$\alpha'$
is $3:3$, as it has eight outgoing edges with labels in~$\Sigma^3$,
and the coding rate at~$\alpha''$ and at~$\beta$ is $2:2$,
as each state has four outgoing edges labeled from $\Sigma^2$;
the coding ratio at each state is therefore~$1$, 
\fi
making~$\encoder$ capacity-achieving.
However, $\encoder$ is not
\ifPAGELIMIT
    deterministic.
\else
deterministic
(there are two edges labeled $bda$ and two labeled $cda$
outgoing from state~$\alpha'$,
two edges labeled $aa$ outgoing from~$\alpha''$,
and two labeled $da$ from state~$\beta$).
Nevertheless, $\encoder$ has finite anticipation
and is therefore lossless:
the first symbol of a label uniquely determines
the length of the label as well as the initial state,
and a label and the first symbol of the next label 
within a sequence uniquely determine the edge.
\fi

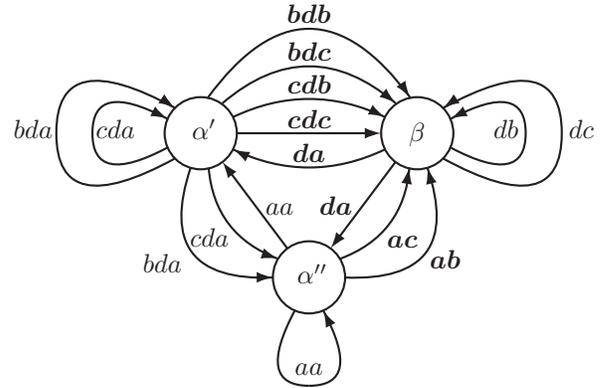
\begin{figure}[hbt]
\begin{center}
\thicklines
\setlength{\unitlength}{\figunit}
\figfont
\ifPAGELIMIT
    \begin{picture}(150,090)(-45,-60)
\else
\begin{picture}(150,110)(-45,-70)
\fi
    \multiput(000,000)(060,000){2}{\circle{20}}
    \put(030,-40){\circle{20}}
    \qbezier(051,4.5)(030,14.5)(009,4.5)
    \put(051,4.5){\vector(2,-1){0}}
    \qbezier(054,8.3)(030,28.8)(006,8.3)
    \put(054,8.3){\vector(5,-4){0}}
    \qbezier(058,9.9)(030,48.0)(002,9.9)
    \put(058,9.9){\vector(2,-3){0}}
    \qbezier(051,-4.5)(030,-14.5)(009,-4.5)
    \put(009,-4.5){\vector(-2,1){0}}

    \put(010,000){\vector(1,0){40}}

    \qbezier(-9.26,004)(-30,015)(-30,000)
    \qbezier(-9.26,-04)(-30,-15)(-30,000)
    \put(-9.26,004){\vector(2,-1){0}}
    \qbezier(-7.14,007)(-40,025)(-40,000)
    \qbezier(-7.14,-07)(-40,-25)(-40,000)
    \put(-7.14,007){\vector(3,-2){0}}

    \qbezier(69.26,004)(090,015)(090,000)
    \qbezier(69.26,-04)(090,-15)(090,000)
    \put(69.26,004){\vector(-2,-1){0}}
    \qbezier(67.14,007)(100,025)(100,000)
    \qbezier(67.14,-07)(100,-25)(100,000)
    \put(67.14,007){\vector(-3,-2){0}}

    \put(024,-32){\vector(-3,4){18}}
    \put(054,-08){\vector(-3,-4){18}}

    \qbezier(040,-40)(73.4,-41.3)(62.8,-9.6)
    \put(62.8,-9.6){\vector(-1,4){0}}
    \qbezier(38.6,-34.8)(56.2,-28.2)(57.8,-9.6)
    \put(57.8,-9.6){\vector(1,6){0}}

    \qbezier(020,-40)(-13.4,-41.3)(-2.8,-9.6)
    \put(020,-40){\vector(1,0){0}}
    \qbezier(21.4,-34.8)(3.8,-28.2)(2.2,-9.6)
    \put(21.4,-34.8){\vector(3,-1){0}}

    \qbezier(034,-49.26)(045,-70)(030,-70)
    \qbezier(026,-49.26)(015,-70)(030,-70)
    \put(034,-49.26){\vector(-1,2){0}}

    \put(001,001){\makebox(0,0){$\alpha'$}}
    \put(060,-.5){\makebox(0,0){$\beta$}}
    \put(031,-39){\makebox(0,0){$\alpha''$}}

    \put(-41,001){\makebox(0,0)[r]{$bda$}}
    \put(-29,001){\makebox(0,0)[l]{$cda$}}
    \put(030,030.6){\makebox(0,0)[b]{$\mathbf{bdb}$}}
    \put(030,020.2){\makebox(0,0)[b]{$\mathbf{bdc}$}}
    \put(030,011.0){\makebox(0,0)[b]{$\mathbf{cdb}$}}
    \put(030,001.5){\makebox(0,0)[b]{$\mathbf{cdc}$}}
    \put(030,-08.0){\makebox(0,0)[b]{$\mathbf{da}$}}
    \put(088,001){\makebox(0,0)[r]{$db$}}
    \put(102,001){\makebox(0,0)[l]{$dc$}}
    \put(042,-20){\makebox(0,0)[r]{$\mathbf{da}$}}
    \put(018,-21){\makebox(0,0)[l]{$aa$}}
    \put(072,-35){\makebox(0,0)[r]{$\mathbf{ab}$}}
    \put(060,-30){\makebox(0,0)[r]{$\mathbf{ac}$}}
    \put(-16,-36){\makebox(0,0)[l]{$bda$}}
    \put(-03,-29){\makebox(0,0)[l]{$cda$}}
    \put(030,-67){\makebox(0,0)[b]{$aa$}}
\end{picture}
\thinlines
\setlength{\unitlength}{1pt}
\end{center}
\caption{VLE~$\encoder$ for
the constraint presented by \Figure~\protect\ref{fig:twostates}.}
\label{fig:vlealt}
\end{figure}

Consider now the following partition~$\{ \Sigma_0, \Sigma_1 \}$
of~$\Sigma$:
\begin{equation}
\label{eq:partitioningalt}
\Sigma_0 = \{ a \} \quad \textrm{and} \quad
\Sigma_1 = \{ b, c, d \}
\ifPAGELIMIT
     \end{equation}
     (the odd labels \withrespectto\ this partition
     are marked in boldface in \Figure~\ref{fig:vlealt}).
     Taking the tag alphabet $\Upsilon = \{ 0, 1 \}$, one
\else
\; .
\end{equation}
\Withrespectto\ this partition, the eight outgoing edges from
state~$\alpha'$ in~$\encoder$ are equally divided between
$(\Sigma^3)_0$ and $(\Sigma^3)_1$,
and the four outgoing edges from
each of the states~$\alpha''$ and~$\beta$ are equally divided between
$(\Sigma^2)_0$ and $(\Sigma^2)_1$
(odd labels are marked in boldface in \Figure~\ref{fig:vlealt}).
Hence, taking the tag alphabet $\Upsilon = \{ 0, 1 \}$,
we can achieve a coding ratio of~$1$ by a parity-preserving VLE.
One
\fi
possible parity-preserving tag assignment
to the edges of~$\encoder$ is shown in Table~\ref{tab:vlealt}.
\begin{table}[hbt]
\caption{Possible tag assignment for
the encoder in \Figure~\ref{fig:vlealt}.}
\label{tab:vlealt}
\ifPAGELIMIT
    \vspace{-2ex}
\fi
\[
\begin{array}{rcccrcccrcc}
\multicolumn{3}{c}{\mathrm{State} \; \alpha'}  &&
\multicolumn{3}{c}{\mathrm{State} \; \alpha''} &&
\multicolumn{3}{c}{\mathrm{State} \; \beta}    \\
\cline{1-3} \cline{5-7} \cline{9-11}
000, 011 & \leftrightarrow   & bda          &\;\;&
00, 11   & \leftrightarrow   & aa           &\;\;&
01, 10   & \leftrightarrow   & \mathbf{da}  \\
101, 110 & \leftrightarrow   & cda          &\;\;&
01       & \leftrightarrow   & \mathbf{ac}  &\;\;&
00       & \leftrightarrow   & db           \\
001      & \leftrightarrow   & \mathbf{bdb} &\;\;&
10       & \leftrightarrow   & \mathbf{ab}  &\;\;&
11       & \leftrightarrow   & dc           \\
010      & \leftrightarrow   & \mathbf{bdc} &&&&&& \\
100      & \leftrightarrow   & \mathbf{cdb} &&&&&& \\
111      & \leftrightarrow   & \mathbf{cdc} &&&&&& \\
\end{array}
\]
\ifPAGELIMIT
    \vspace{-2ex}
\fi
\end{table}
Similarly to the partition~(\ref{eq:partitioning}),
it was shown in~\cite{RS1}
that for the partition~(\ref{eq:partitioningalt}), too,
one cannot achieve a coding ratio of~$1$
by any parity-preserving fixed-length encoder
\ifPAGELIMIT
    for~$S$.
    Yet unlike the partition~(\ref{eq:partitioning}),
    for~(\ref{eq:partitioningalt}), it can be shown (details omitted)
    that there is no deterministic VLE that has coding ratio~$1$
    and a parity-preserving assignment.\qed
\else
for~$S$.

The encoder in \Figure~\ref{fig:vlealt}
can be obtained from (an untagged copy of)
the encoder in \Figure~\ref{fig:fixedlength}
by replacing the outgoing edges from state~$\alpha'$
with the eight paths of length~$3$ that start at that state
and, similarly, replacing the outgoing edges from
each of the states~$\alpha''$ and~$\beta$ with the four paths of
length~$2$ that start at the state.\qed
\fi
\end{example}

\ifPAGELIMIT
\else
To summarize, for the constraint~$S$ of Example~\ref{ex:twostates},
Examples~\ref{ex:vle} and~\ref{ex:vlealt} present, respectively,
(capacity-achieving) parity-preserving VLEs with a coding ratio of~$1$
for the
two partitions~(\ref{eq:partitioning}) and~(\ref{eq:partitioningalt}):
the first VLE is deterministic, while the other is not.
In fact, we show in Appendix~\ref{sec:nonexistence}
that for the partition~(\ref{eq:partitioningalt}),
one cannot achieve a coding ratio of~$1$ by
any deterministic parity-preserving VLE
(unless one uses a degenerate base tag alphabet
containing only even symbols).
On the other hand, there exists such an encoder
under some relaxation of the notion of fixed coding ratio,
following the encoding model considered in~\cite{HMS}:
the tagged encoder~$\encoder^\circ$ in \Figure~\ref{fig:vleHMS}
maintains a coding ratio of~$1$ \emph{along each cycle}.
\begin{figure}[hbt]
\begin{center}
\thicklines
\setlength{\unitlength}{\figunit}
\figfont
\begin{picture}(150,050)(-45,-25)
    \multiput(000,000)(060,000){2}{\circle{20}}
    \qbezier(051,4.5)(030,14.5)(009,4.5)
    \put(051,4.5){\vector(2,-1){0}}
    \qbezier(051,-4.5)(030,-14.5)(009,-4.5)
    \put(009,-4.5){\vector(-2,1){0}}
    \qbezier(054,-8.3)(030,-28.8)(006,-8.3)
    \put(006,-8.3){\vector(-5,4){0}}

    \qbezier(-9.26,004)(-30,015)(-30,000)
    \qbezier(-9.26,-04)(-30,-15)(-30,000)
    \put(-9.26,004){\vector(2,-1){0}}
    \qbezier(-7.14,007)(-40,025)(-40,000)
    \qbezier(-7.14,-07)(-40,-25)(-40,000)
    \put(-7.14,007){\vector(3,-2){0}}

    \qbezier(69.26,004)(090,015)(090,000)
    \qbezier(69.26,-04)(090,-15)(090,000)
    \put(69.26,004){\vector(-2,-1){0}}

    \put(000,000){\makebox(0,0){$\alpha$}}
    \put(060,-.5){\makebox(0,0){$\beta$}}

    \put(-43,000){\makebox(0,0)[r]{$11/cd$}}
    \put(-27,000){\makebox(0,0)[l]{$0/a$}}
    \put(030,011.0){\makebox(0,0)[b]{$10/\mathbf{b}$}}
    \put(030,-08.0){\makebox(0,0)[b]{$1/\mathbf{da}$}}
    \put(030,-20.0){\makebox(0,0)[t]{$01/\mathbf{dcd}$}}
    \put(093,000){\makebox(0,0)[l]{$00/db$}}
\end{picture}
\thinlines
\setlength{\unitlength}{1pt}
\end{center}
\caption{Second VLE~$\encoder^\circ$ for
the constraint presented by \Figure~\protect\ref{fig:twostates}.}
\label{fig:vleHMS}
\end{figure}
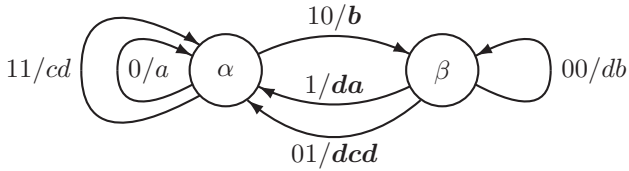
It is easily seen that while at state~$\alpha$,
each outgoing edge is uniquely determined by its first symbol,
and while at state~$\beta$, an outgoing edge is
uniquely determined by its first two symbols.

\begin{remark}
\label{rem:codingratio}
Extending the terminology from fixed-length encoders,
in a tagged VLE at a (fixed) coding ratio~$p/q$
for a constraint~$S$, input tags are words over the (base) tag alphabet,
and the length of a tag of each edge equals $p/q$ times the edge length.
The set of tags of the outgoing edges from each state must form
an exhaustive prefix-free list.
Assuming that $\gcd(p,q) = 1$, the length~$\ell$ of an edge must be
divisible by~$q$, so we can consider the constraint~$S^q$ instead
and regard each length-$\ell$ label over~$\Sigma$
as a word of length $\ell/q$ over $\Sigma^q$.
Accordingly, we can group the $p \ell /q$ symbols in each tag
into $\ell/q$ blocks of length~$p$.
Doing so, the coding ratio becomes~$1$.\qed
\end{remark}
\fi

\begin{example}
\label{ex:2inf}
Let~$S$ be the $(2,\infty)$-RLL constraint, whose Shannon cover
is given by the graph~$G$ in \Figure~\ref{fig:2inf}.
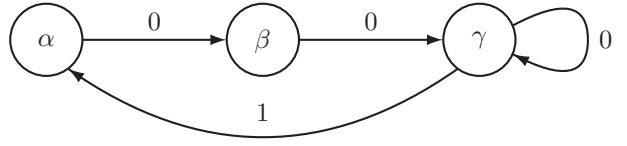
\begin{figure}[hbt]
\begin{center}
\thicklines
\setlength{\unitlength}{\figunit}
\figfont
\ifPAGELIMIT
    \begin{picture}(170,030)(-10,-20)
\else
\begin{picture}(170,050)(-10,-30)
\fi
    \multiput(000,000)(060,000){3}{\circle{20}}
    \put(010,000){\vector(1,0){40}}
    \put(030,005){\makebox(0,0){$0$}}
    \put(070,000){\vector(1,0){40}}
    \put(090,005){\makebox(0,0){$0$}}
    \qbezier(129.26,004)(150,015)(150,000)
    \qbezier(129.26,-04)(150,-15)(150,000)
    \put(155,000){\makebox(0,0){$0$}}
    \put(129.26,-04){\vector(-2,1){0}}
    \qbezier(114,-08)(060,-46)(006,-08)
    \put(006,-08){\vector(-3,2){0}}
    \put(060,-20){\makebox(0,0){$1$}}
    \put(000,000){\makebox(0,0){$\alpha$}}
    \put(060,-.5){\makebox(0,0){$\beta$}}
    \put(120,000){\makebox(0,0){$\gamma$}}
\end{picture}
\thinlines
\setlength{\unitlength}{1pt}
\end{center}
\caption{Shannon cover~$G$ of the $(2,\infty)$-RLL constraint.}
\label{fig:2inf}
\end{figure}
The capacity of~$S$ is approximately $0.5515$, so there exists
a rate $1:2$ fixed-length encoder for~$S$
(namely, an $(S^2,2)$-encoder); such a (tagged) encoder~$\encoder$
is shown in
\ifPAGELIMIT
    \Figure~\ref{fig:2infencoder}.
\else
\Figure~\ref{fig:2infencoder}
(note that in this case, $S(\encoder)$ is strictly contained in $S^2$).
\fi
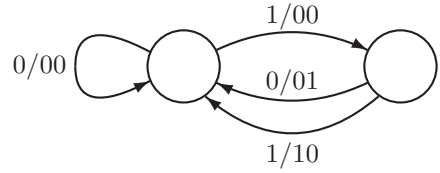
\begin{figure}[hbt]
\begin{center}
\thicklines
\setlength{\unitlength}{\figunit}
\figfont
\ifPAGELIMIT
    \begin{picture}(125,030)(-45,-15)
\else
\begin{picture}(125,050)(-45,-25)
\fi
    \multiput(000,000)(060,000){2}{\circle{20}}
    \qbezier(051,4.5)(030,14.5)(009,4.5)
    \put(051,4.5){\vector(2,-1){0}}
    \qbezier(051,-4.5)(030,-14.5)(009,-4.5)
    \put(009,-4.5){\vector(-2,1){0}}
    \qbezier(054,-8.3)(030,-28.8)(006,-8.3)
    \put(006,-8.3){\vector(-5,4){0}}
    \put(030,010.0){\makebox(0,0)[b]{$1/00$}}
    \put(030,-08.0){\makebox(0,0)[b]{$0/01$}}
    \put(030,-21.0){\makebox(0,0)[t]{$1/10$}}
    \qbezier(-9.26,004)(-30,015)(-30,000)
    \qbezier(-9.26,-04)(-30,-15)(-30,000)
    \put(-40,000){\makebox(0,0){$0/00$}}
    \put(-9.26,-04){\vector(2,1){0}}
\end{picture}
\thinlines
\setlength{\unitlength}{1pt}
\end{center}
\caption{Rate $1:2$ fixed-length encoder~$\encoder$ for
the $(2,\infty)$-RLL constraint.}
\label{fig:2infencoder}
\end{figure}
This encoder is not deterministic;
in fact, the smallest integer~$p$ for which there exists
a rate $p:2p$ deterministic fixed-length encoder for~$S$
is $p = 7$, as this is the smallest integer for which
the set $\XX(A_G^{2p},2^p)$ contains a $0\mbox{--}1$ vector
(see~\cite[Theorem~7.15]{MRS}).
\ifPAGELIMIT
\else
Still, the encoder~$\encoder$ is $(0,1)$-sliding-block decodable.
\fi

On the other hand, the graph in \Figure~\ref{fig:2infvle},
with the tagging
\ifPAGELIMIT
    $0  \leftrightarrow 00$,
    $10 \leftrightarrow 01{.}00$, and
    $11 \leftrightarrow 10{.}00$,
\else
of Table~\ref{tab:2infvle},
\fi
is a deterministic VLE for~$S$ with a coding ratio of $1/2$
(see~\cite{Franaszek2};
since the alphabet of~$S^2$ consists of pairs of bits, 
we have used dots to delimit the symbols within each label).
Note, however, that
\ifPAGELIMIT
    such a tag assignment
\else
the tag assignment in Table~\ref{tab:2infvle}
\fi
is \emph{not} parity-preserving;
we will return to this example in
Examples~\ref{ex:2infprincipalstates}
and~\ref{ex:2infppprincipalstates} below.\qed
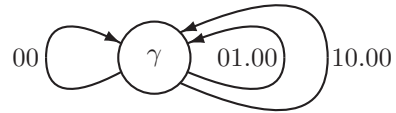
\begin{figure}[hbt]
\begin{center}
\thicklines
\setlength{\unitlength}{\figunit}
\figfont
\ifPAGELIMIT
    \begin{picture}(98,015)(-35,-05)
\else
\begin{picture}(98,035)(-35,-15)
\fi
    \put(000,000){\circle{20}}

    \qbezier(-9.26,004)(-30,015)(-30,000)
    \qbezier(-9.26,-04)(-30,-15)(-30,000)
    \put(-9.26,004){\vector(2,-1){0}}

    \qbezier(9.26,004)(036,015)(036,000)
    \qbezier(9.26,-04)(036,-15)(036,000)
    \put(9.26,004){\vector(-3,-1){0}}
    \qbezier(7.14,007)(048,025)(048,000)
    \qbezier(7.14,-07)(048,-25)(048,000)
    \put(7.14,007){\vector(-5,-2){0}}

    \put(000,000){\makebox(0,0){$\gamma$}}
    \put(-32,000){\makebox(0,0)[r]{$00$}}
    \put(034,000){\makebox(0,0)[r]{$01{.}00$}}
    \put(049,000){\makebox(0,0)[l]{$10{.}00$}}
\end{picture}
\thinlines
\setlength{\unitlength}{1pt}
\end{center}
\caption{VLE for the $(2,\infty)$-RLL constraint.}
\label{fig:2infvle}
\end{figure}
\ifPAGELIMIT
\else
\begin{table}[hbt]
\caption{Possible tag assignment for the encoder in
\Figure~\ref{fig:2infvle}.}
\label{tab:2infvle}
\[
\begin{array}{lcl}
0  & \leftrightarrow & 00      \\
10 & \leftrightarrow & 01{.}00 \\
11 & \leftrightarrow & 10{.}00 \\
\end{array}
\]
\end{table}
\fi
\end{example}

\subsection{Deterministic variable-length encoders}
\label{sec:DVLE}

In this section, we focus on VLEs which are deterministic,
and quote a necessary and sufficient condition for having
such encoders. 

Let $H = (V,E,L)$ be a VLG whose labels are over
a finite alphabet~$\Sigma$ and let~$n$ be a positive integer.
Fix some nonempty subset $V' \subseteq V$,
and let $H' = (V',E',L')$ be the subgraph of~$H$ that is
induced by~$V'$ (namely, $E'$ consists of all the edges in~$H$
both of whose endpoints are in~$V'$).
For every $u \in V'$ and $\ell \ge 1$, denote
by $\mu_\ell(u|V')$ the number
of outgoing edges of length~$\ell$ from~$u$ in $H'$.
We say that~$V'$ is
a set of \emph{principal states in~$H$} \withrespectto~$n$ if
for every $u \in V'$:
\ifPAGELIMIT
    \[
    \sum_{\ell \ge 1} \frac{\mu_\ell(u|V')}{n^\ell} \ge 1 \; .
    \]
\else
\begin{equation}
\label{eq:Kraftinequality}
\sum_{\ell \ge 1} \frac{\mu_\ell(u|V')}{n^\ell} \ge 1 \; .
\end{equation}
\fi
It readily follows from this definition
that~$V'$ is a set of principal states in a VLG~$H$
\withrespectto~$n$, \ifandonlyif\ it is also so in
the subgraph~$H'$ of~$H$ that is induced by~$V'$.

The following result is essentially known
(see~\cite{Beal1},
\cite{Franaszek1},
\cite{Franaszek2}).

\begin{theorem}
\label{thm:principalstates}
Let~$S$ be an irreducible constraint and
let~$n$ and~$r$ be positive integers.
There exists a deterministic $(S,n)$-VLE
whose edges all have length at most~$r$,
\ifandonlyif\ $S$
is presented by an irreducible deterministic VLG $H = (V,E,L)$
whose edges all have length at most~$r$,
and~$V$ contains a subset of principal states
\withrespectto~$n$.\footnote{%
Moreover, the graph~$H$ can be assumed to be \emph{reduced}, namely,
the follower sets of its states are distinct.
For the case where all the edge lengths are~$1$,
the graph~$H$ is the Shannon cover of~$S$.}
\end{theorem}

\ifPAGELIMIT
\else
We include a proof of the theorem
both for completeness and for reference
in our upcoming extension of this result to the parity-preserving case.

\begin{proof}[Proof of Theorem~\ref{thm:principalstates}]
Sufficiency follows by first looking at the subgraph~$H'$
of~$H$ that is induced by a set of principal states~$V'$.
We then (possibly) remove outgoing edges from states in~$H'$, 
starting with the longest outgoing edge and proceeding
(if necessary) with edges in descending order of their lengths,
until the inequality~(\ref{eq:Kraftinequality}) becomes
an equality at each state $u \in V'$.

To show necessity, suppose that~$\encoder$ is
a deterministic $(S,n)$-VLE. By shifting to an irreducible
sink\footnote{%
\label{fn:sink}%
An irreducible sink of $\encoder$ is an irreducible subgraph
$\overline{\encoder} = (\overline{V},\overline{E},\overline{L})$
of $\encoder$ such that all the outgoing edges from $\overline{V}$
in $\encoder$ terminate in $\overline{V}$.
Every graph has at least one irreducible sink~\cite[\S 2.5.1]{MRS}.
It is straightforward to see that an irreducible sink of
an $(S,n)$-VLE is also an $(S,n)$-VLE.}
of~$\encoder$, we can assume that~$\encoder$ is irreducible.
Let~$G$ be the Shannon cover of~$S$.
By transforming the outgoing edges from each state in~$\encoder$
into a tree (as in Remark~\ref{rem:vlg}),
we get from~\cite[Lemma~2.13]{MRS}
that for every state~$Z$ in~$\encoder$
there is a state~$u$ in~$G$ such that
$\FF_\encoder(Z) \subseteq \FF_G(u)$.
Let $V'$ denote the states in~$G$ whose follower sets contain
follower sets of states of~$\encoder$; clearly,
$V'$ is not empty.
For every $u \in V'$, let $Z(u)$ be
some particular state~$Z$ in~$\encoder$ such that
$\FF_\encoder(Z) \subseteq \FF_G(u)$.

Next, we construct a deterministic VLG $H^* = (V',E^*,L^*)$
in which the outgoing edges from each state $u \in V'$ are
defined as follows: for each edge
$Z(u) \stackrel{\bldw}{\rightarrow} \tilde{Z}$
outgoing from $Z(u)$ in~$\encoder$
(where $\tilde{Z}$ is the terminal state of the edge
and $\bldw$ is its label), endow~$H^*$ with
an edge $u \stackrel{\bldw}{\rightarrow} \tilde{u}$,
where~$\tilde{u}$ is the terminal state
of the (unique) path in~$G$ that starts at~$u$
and generates the word~$\bldw$;
note that $\FF_\encoder(\tilde{Z}) \subseteq \FF_G(\tilde{u})$
and, therefore, $\tilde{u} \in V'$.
By the construction it follows that~$H^*$ is deterministic and,
by possibly shifting to an irreducible sink of~$H^*$,
we can assume that~$H^*$ is irreducible.
One can easily show by induction on~$\ell$
that all length-$\ell$ words in $\FF_{H^*}(u)$
are contained in $\FF_G(u)$, for every $u \in V'$;
hence, $\FF_{H^*}(u) \subseteq \FF_G(u)$ for every $u \in V'$
and, in particular, $S(H^*) \subseteq S(G) = S$.
Moreover, denoting by $\mu^*_\ell(u)$
the number of outgoing edges of length~$\ell$ from
state~$u$ in~$H^*$, we have, for every $u \in V'$:
\begin{equation}
\label{eq:Kraftequality}
\sum_{\ell \ge 1} \frac{\mu^*_\ell(u)}{n^\ell} = 1 \; .
\end{equation}
Thus, $H^*$ is an irreducible deterministic $(S,n)$-VLE.
Moreover, the length of each edge in~$H^*$
is at most the length of the longest edge in~$\encoder$.

Next, we construct a VLG $H = (V,E,L)$
that contains~$H^*$ as a subgraph, as follows.
Start with $(V,E,L) \leftarrow (V',E^*,L^*)$.
Then, for each state $u \in V'$, let $\length(u)$ be
the length of the longest edge outgoing from~$u$ in~$H^*$.
For every word $\bldw \in \FF_G(u) \setminus \FF_{H^*}(u)$
of length~$\length(u)$ that does not have any prefix
that labels any of the outgoing edges from~$u$ in~$H^*$,
endow~$H$ with an edge labeled~$\bldw$ from~$u$
to the terminal state~$v$ of the path from~$u$ in~$G$
that generates~$\bldw$
(in particular, insert~$v$ into~$V$ if it is not there already).
Finally, iteratively endow~$H$ with the (length-$1$) outgoing edges
(in~$G$) from each state $u \in V \setminus V'$
(and insert their terminal states to~$V$ if they are not there already),
until no new edges are added.

We claim that~$H$ is irreducible.
Indeed, the subgraph~$H^*$
is irreducible, and every state $u \in V \setminus V'$
is reachable from~$V'$ in~$H$ (or else it would not have been
inserted into~$V$). Moreover, from each state
$u \in V \setminus V'$ we can reach some state
in~$V'$ in~$H$ by following the shortest path from~$u$
to~$V'$ in the Shannon cover~$G$.

Secondly, we claim that~$H$ is deterministic.
Indeed, at each state~$u \in V'$
we only add edges of length~$\length(u)$
whose labels do not have prefixes that label
the existing outgoing edges from~$u$ in~$H^*$,
and at each state $u \in V \setminus V'$,
the outgoing neighborhood from~$u$ in~$H$ is the same
as that in~$G$.

Thirdly, we claim that $\FF_H(u) = \FF_G(u)$ for every $u \in V$.
We prove this by induction, showing that 
$\FF_H(u) \cap \Sigma^\ell = \FF_G(u) \cap \Sigma^\ell$
for every $\ell \ge 0$. The induction base $\ell = 0$ is trivial,
due to the empty word. As for the induction step,
the case $u \in V \setminus V'$ is immediate,
while the case $u \in V'$ follows from the addition
of the edges labeled by
words $\bldw \in \left( \FF_G(u) \setminus \FF_{H^*}(u) \right)
\cap \Sigma^{\length(u)}$
(whose prefixes do not label outgoing edges from~$u$) to~$H$.
Irreducibility of both~$G$ and~$H$ 
and the equality $\FF_H(u) = \FF_G(u)$
(for some state $u \in V$) then imply that $S(H) = S(G)$.

Finally, since $H^*$ is a subgraph of the subgraph~$H'$ of~$H$
that is induced by~$V'$, we get from~(\ref{eq:Kraftequality})
that~(\ref{eq:Kraftinequality}) holds for every $u \in V'$, namely,
$V'$ is a set of principal states in~$H$ \withrespectto~$n$.
\end{proof}

\begin{remark}
\label{rem:principalstates1}
It follows from Remark~\ref{rem:vlg}
that when~$H$ is (irreducible, deterministic, and) reduced,
its set of states is in effect a subset of the set of states of
the Shannon cover~$G$ of $S(H)$. Therefore,
any principal set of states~$V'$ of such an~$H$
consists of states of the Shannon cover of~$S(H)$.\qed
\end{remark}

\begin{remark}
\label{rem:principalstates2}
It follows from the proof of the ``if'' part 
of Theorem~\ref{thm:principalstates} that 
if an irreducible deterministic VLG~$H$ contains
a set~$V'$ of principal states \withrespectto~$n$,
then there is a deterministic
$(S,n)$-VLE $\encoder = (V',\tilde{E},\tilde{L})$
which is a subgraph of the subgraph $H'$ of~$H$
that is induced by~$V'$. Moreover, $V'$ can be assumed
to be the set of states of~$\encoder$ (although~$\encoder$
could then be reducible).\qed
\end{remark}
\fi

Given an ordinary irreducible deterministic graph~$G$
(with length-$1$ edges) and positive integers~$n$ and~$r$,
Franaszek described in~\cite{Franaszek2} a polynomial-time algorithm
for testing whether~$S(G)$ can be presented by a VLG~$H$ 
that satisfies the conditions of Theorem~\ref{thm:principalstates}
(see also \cite{Beal1}, \cite{Beal2}).
His algorithm, which is based on dynamic programming,
effectively finds a set of principal states~$V'$
(which is a subset of the states of~$G$)
and a subgraph~$H'$ of~$H$ that is induced by~$V'$
(the graph~$H$ itself is not explicitly constructed
in~\cite{Franaszek2}).

\begin{example}
\label{ex:2infprincipalstates}
\ifPAGELIMIT
    Let~$G$ and~$S$ be as in Example~\ref{ex:2inf},
\else
Let~$S$ be the $(2,\infty)$-RLL constraint,
which is presented by the graph~$G$ in \Figure~\ref{fig:2inf},
\fi
and take $n = 2$. Since there are no deterministic $(S^2,2)$-encoders,
we cannot have any principal states when $r = 1$.
\ifPAGELIMIT
\else

\fi
Selecting $r = 2$, an application
of Franaszek's algorithm from~\cite{Franaszek2} to~$G^2$
yields a (unique) set of principal
states~$V'$ consisting only of state $\gamma$.
Since \withoutlossofgenerality~$H$ is reduced, that implies
a unique subgraph~$H'$ that is induced by~$V'$,
which is the graph in \Figure~\ref{fig:2infvle}
(see~\cite[\S V]{Franaszek2}).\qed
\end{example}

\ifPAGELIMIT
    \section{Parity-preserving variable-length encoders}
    \label{sec:PPVLE}
    In this section,
\else
\section{Parity-preserving Kraft conditions}
\label{sec:Kraft}

In Section~\ref{sec:PPVLE},
\fi
we provide a formal definition of
a parity-preserving variable-length encoder.
A key ingredient in that definition
will be an adaptation of Theorem~\ref{thm:Kraft}
to the parity-preserving case, which we do
\ifPAGELIMIT
    in \Section~\ref{sec:Kraft};
\else
    next;
\fi
that adaptation may be of independent interest,
beyond its use in this work.
\ifPAGELIMIT
    We then state a necessary and sufficient condition
    for having a parity-preserving VLE which is deterministic.

    \subsection{Parity-preserving Kraft conditions}
    \label{sec:Kraft}
\else
The main result of this section is
Theorem~\ref{thm:lengthdist} below, whose statement
uses the following definition and notation.
\fi

Let~$\Upsilon$ be a finite alphabet
and assume a partition $\{ \Upsilon_0, \Upsilon_1 \}$ of~$\Upsilon$.
Given a finite list~$\List$ of nonempty words over~$\Upsilon$,
the \emph{(parity-preserving) length distribution} of~$\List$
is a pair of nonnegative integer sequences
$\left( \bldeta{=}(\eta_\ell)_{\ell \ge 1},
\bldomega{=}(\omega_\ell)_{\ell \ge 1} \right)$,
where
\[
\eta_\ell = \left| \List \cap (\Upsilon^\ell)_0 \right|
\;\; \textrm{and} \;\;
\omega_\ell = \left| \List \cap (\Upsilon^\ell)_1 \right|
\; , \;\;
\ell = 1, 2, 3, \cdots
\; .
\]
In words, $\eta_\ell$ (\respectively, $\omega_\ell$)
is the number of even (\respectively, odd) length-$\ell$
words in $\List$.

Given integers~$n$ and~$\ell > 0$
and an integer sequence $\bldmu = (\mu_i)_{i \ge 1}$
with finite support, we define the following functional:
\[
\Kraft_\ell(\bldmu,n)
= n^\ell - \sum_{i=1}^\ell \mu_i \cdot n^{\ell-i} \; .
\]
Given now positive integers~$n_0$, $n_1$, and~$\ell$
and a pair $(\bldeta,\bldomega)$ of nonnegative integer sequences,
each with finite support, define
\[
\Kraft_\ell^+ = \Kraft_\ell(\bldeta + \bldomega,n_0 + n_1)
\; \phantom{.}
\]
and
\[
\Kraft_\ell^- = \Kraft_\ell(\bldeta - \bldomega,n_0 - n_1) \; .
\]
Thus,
\ifPAGELIMIT
    \begin{eqnarray*}
    \Kraft_\ell^\pm 
    & = & \Kraft_\ell(\bldeta \pm \bldomega,n_0 \pm n_1) \nonumber \\
    & = &
    (n_0 \pm n_1)^\ell -
    \sum_{i=1}^\ell (\eta_i \pm \omega_i)(n_0 \pm n_1)^{\ell-i}
    \; .
    \end{eqnarray*}
    Denoting
\else
\begin{eqnarray}
\Kraft_\ell^\pm 
& = & \Kraft_\ell(\bldeta \pm \bldomega,n_0 \pm n_1) \nonumber \\
& = &
(n_0 \pm n_1)^\ell -
\sum_{i=1}^\ell (\eta_i \pm \omega_i)(n_0 \pm n_1)^{\ell-i} \nonumber \\
\label{eq:Kne}
& = &
(n_0 \pm n_1)^\ell \cdot
\left(
1 -
\sum_{i=1}^\ell \frac{\eta_i \pm \omega_i}{(n_0 \pm n_1)^i}
\right)
\; ,
\end{eqnarray}
where the last equality applies for $\Kraft_\ell^-$
only when $n_0 \ne n_1$; when $n_0 = n_1$ we have instead:
\begin{equation}
\label{eq:Keq}
\Kraft_\ell^- = \omega_\ell - \eta_\ell \; .
\end{equation}
Denoting hereafter
\fi
by $\length = \length(\bldeta,\bldomega)$ the largest index
in the union of the supports of~$\bldeta$ and~$\bldomega$,
the notation $\Kraft^\pm = \Kraft(\bldeta \pm \bldomega,n_0 \pm n_1)$
will stand
\ifPAGELIMIT
    for
\else
for\footnote{%
There is a slight abuse in
the notation $\Kraft(\bldeta - \bldomega,n_0 - n_1)$,
since sometimes $\length(\bldeta,\bldomega)$
is not uniquely determined from $\bldeta - \bldomega$.}
\fi
$\Kraft_\length^\pm$.
Thus, (\ref{eq:Kraft}) becomes
\ifPAGELIMIT
    \[
    \Kraft^+ = \Kraft^+(\bldeta + \bldmu, n_0 + n_1) = 0 \; ,
    \]
\else
\begin{equation}
\label{eq:Kplus}
\Kraft^+ = \Kraft^+(\bldeta + \bldmu, n_0 + n_1) = 0 \; ,
\end{equation}
\fi
where we have taken $n_0 = |\Upsilon_0|$ and $n_1 = |\Upsilon_1|$.

The next theorem provides a necessary and sufficient condition
for a pair $(\bldeta,\bldomega)$ to be a (parity-preserving)
length distribution of an exhaustive prefix-free list.

\begin{theorem}
\label{thm:lengthdist}
Given a partition $\{ \Upsilon_0, \Upsilon_1 \}$
of a finite alphabet $\Upsilon$
with $|\Upsilon_0| = n_0$ and $|\Upsilon_1| = n_1$,
let $(\bldeta,\bldomega)$ be
a pair of nonnegative integer sequences, each with finite support.
Then there exists an exhaustive prefix-free list
over~$\Upsilon$ with a length distribution $(\bldeta,\bldomega)$,
\ifandonlyif\ the following conditions hold.
\begin{list}{}{\settowidth{\labelwidth}{\textit{(a)}}%
               \settowidth{\leftmargin}{\textit{(a)}.....}}
\item[(a)]
$\displaystyle \Kraft^+ = 0$, and---
\item[(b)]
\ifPAGELIMIT
    $\Kraft^+_\ell \ge \left| \Kraft_\ell^- \right|$
    for every $\ell \ge 1$.
\else
for every $\ell \ge 1$:
\begin{equation}
\label{eq:conditionb}
\Kraft^+_\ell \ge \left| \Kraft_\ell^- \right| \; .
\end{equation}
\fi
\end{list}
\end{theorem}

\ifPAGELIMIT
    The proof of the proposition can be found in~\cite{RS2}.
\else
\begin{remark}
\label{rem:lengthdist}
For $\ell \ge \length = \length(\bldeta,\bldomega)$
we have $\Kraft_\ell^\pm = (n_0 \pm n_1)^{\ell-\length} K^\pm$;
hence, condition~(a) is equivalent to requiring that $\Kraft_\ell^+ = 0$
for \emph{any} $\ell \ge \length$.
Conditioning on~(a), the inequality~(\ref{eq:conditionb})
for $\ell = \length$ is equivalent to
\begin{equation}
\label{eq:Kpm}
\Kraft^+ = \Kraft^- = 0 \; ,
\end{equation}
so it suffices to state condition~(b) 
only for $1 \le \ell \le \length$: for larger~$\ell$,
the inequality~(\ref{eq:conditionb}) follows from~(\ref{eq:Kpm})
(and holds with equality).\qed
\end{remark}

We prove Theorem~\ref{thm:lengthdist}
through a sequence of intermediary results, starting with
the following equivalent formulation of
conditions~(a) and~(b) (which is somewhat more explicit).

\begin{lemma}
\label{lem:lengthdist}
Conditions~(a) and~(b) in Theorem~\ref{thm:lengthdist}
are equivalent to the following conditions.
\begin{list}{}{\settowidth{\labelwidth}{\textit{(iii)}}}
\item[{\textit{(i)}}]
$\displaystyle
\sum_{\ell \ge 1} \frac{\eta_\ell + \omega_\ell}{(n_0 + n_1)^\ell} = 1
\; ,
$
\item[{\textit{(ii)}}]
$\displaystyle
\sum_{\ell \ge 1} \frac{\eta_\ell - \omega_\ell}{(n_0 - n_1)^\ell} = 1
\; ,
$
whenever $n_0 \ne n_1$, and---
\item[{\textit{(iii)}}]
for every $\ell \ge 1$:
\begin{equation}
\label{eq:lengthdist}
\sum_{i\ge 1} \frac{\eta_{\ell+i} + \omega_{\ell+i}}{(n_0 + n_1)^i}
\ge
\left\{
\renewcommand{\arraystretch}{1.8}
\begin{array}{ll}
\left| \eta_\ell - \omega_\ell \right|
& \textrm{if $n_0 = n_1$} \\
\displaystyle
\Bigl|
\sum_{i \ge 1} \frac{\eta_{\ell+i} - \omega_{\ell+i}}{(n_0 - n_1)^i}
\Bigr|
& \textrm{if $n_0 \ne n_1$} \\
\end{array}
\right.
\! .
\end{equation}
\end{list}
\end{lemma}

\begin{proof}
Clearly, conditions~(a) and~(i) are equivalent.
Next, we observe that for $n_0 = n_1$,
the inequality~(\ref{eq:lengthdist})
implies that $\eta_\length = \omega_\length$
for $\length = \length(\bldeta,\bldomega)$.
Hence, the following restatement of condition~(ii) does not
effectively change conditions (i)---(iii):
\begin{list}{}{\settowidth{\labelwidth}{\textit{(ii')}}}
\item[{\textit{(ii')}}]
$\displaystyle
\left\{
\renewcommand{\arraystretch}{1.8}
\begin{array}{ll}
\displaystyle
\omega_\length - \eta_\length = 0
& \textrm{if $n_0 = n_1$} \\
\displaystyle
\sum_{\ell \ge 1} \frac{\eta_\ell - \omega_\ell}{(n_0 - n_1)^\ell} = 1
& \textrm{if $n_0 \ne n_1$}
\end{array}
\right.
\; .
$
\end{list}
By~(\ref{eq:Kne}) and~(\ref{eq:Keq})
it follows that conditions~(i) and~(ii') are equivalent to
requiring $\Kraft^+ = \Kraft^- = 0$.
Moreover, conditioning on~(i) and~(ii')
(or conditioning on $0 = \Kraft^+ \ge \left| \Kraft^- \right|$), we have
\[
\sum_{i=1}^\ell \frac{\eta_i \pm \omega_i}{(n_0 \pm n_1)^i}
+
\sum_{i \ge 1} \frac{\eta_{\ell+i}
\pm \omega_{\ell+i}}{(n_0 \pm  n_1)^{i+\ell}}
= 1 \; .
\]
Therefore,
\begin{eqnarray*}
\sum_{i \ge 1} \frac{\eta_{\ell+i} \pm\omega_{\ell+i}}{(n_0 \pm  n_1)^i}
& = &
(n_0 \pm n_1)^\ell 
\left(
1 -
\sum_{i=1}^\ell \frac{\eta_i \pm \omega_i}{(n_0 \pm n_1)^i}
\right)  \\
& = &
\Kraft_\ell^\pm \; ,
\end{eqnarray*}
and, so, (\ref{eq:lengthdist}) is equivalent to
\[
\Kraft_\ell^+ \ge \left| \Kraft_\ell^- \right|
\]
(even when $n_0 = n_1$).
We conclude that conditions (i)--(iii)
are equivalent to conditions (a)--(b).
\end{proof}

\begin{lemma}
\label{lem:prefixdist}
Given a partition $\{ \Upsilon_0, \Upsilon_1 \}$
of a finite alphabet~$\Upsilon$
with $|\Upsilon_0| = n_0$ and $|\Upsilon_1| = n_1$,
let $(\bldeta,\bldomega)$ be
a pair of nonnegative integer sequences, each with finite support.
Then there exists an exhaustive prefix-free list
over~$\Upsilon$ with a length distribution $(\bldeta,\bldomega)$,
\ifandonlyif\ there exists a pair of nonnegative integer sequences
$\left(\bldy{=}(y_\ell)_{\ell\ge 0},\bldz{=}(z_\ell)_{\ell\ge 0}\right)$
with finite support such that for every $\ell \ge 1$:
\begin{equation}
\label{eq:prefixdist}
\renewcommand{\arraystretch}{1.2}
\begin{array}{rcl}
\displaystyle
\eta_\ell   & = & n_0 y_{\ell-1} + n_1 z_{\ell-1} - y_\ell \\
\displaystyle
\omega_\ell & = & n_1 y_{\ell-1} + n_0 z_{\ell-1} - z_\ell
\; ,
\end{array}
\end{equation}
where $y_0 \equiv 1$ and $z_0 \equiv 0$.
\end{lemma}

\begin{proof}
We start with proving the ``only if'' part.
Let $(\bldeta,\bldomega)$ be the length distribution
of an exhaustive prefix-free list~$\List$, and let~$\Prefix$ denote
the set of words over~$\Upsilon$
which are \emph{proper} prefixes of words in~$\List$;
namely, a word~$\bldw$ is in~$\Prefix$
\ifandonlyif\ there exists a nonempty word~$\bldw'$ over~$\Upsilon$
such that $\bldw \bldw' \in \List$
(in particular, $\Prefix$ always contains the empty word).
Since $\List$ is prefix-free, it cannot contain any
of the (not necessarily proper) prefixes of the words in $\Prefix$;
in particular, $\List \cap \Prefix = \emptyset$.
On the other hand, since~$\List$ is exhaustive, 
for any $s \in \Upsilon$
and $\bldw \in \Prefix$,
either $\bldw s \in \List$
or $\bldw s \in \Prefix$ (but not both).
Hence,
\[
\left\{ \bldw s \;:\: 
s \in \Upsilon, \bldw \in \Prefix \right\} =
\List  \cup \Prefix
\]
and, so, for every $\ell \ge 1$ and $\parity \in \{ 0, 1 \}$:
\begin{equation}
\label{eq:prefix}
\left\{ \bldw s \;:\: 
s \in \Upsilon, \bldw \in \Prefix \right\} \cap (\Upsilon^\ell)_\parity
= \left( \List \cap (\Upsilon^\ell)_\parity \right)
\cup
\left( \Prefix \cap (\Upsilon^\ell)_\parity \right) .
\end{equation}
For every $\ell \ge 0$,
let $y_\ell$ (\respectively, $z_\ell$) denote the number of
length-$\ell$ even (\respectively, odd) words in $\Prefix$:
\begin{eqnarray*}
y_\ell   & = & \left| \Prefix \cap (\Upsilon^\ell)_0 \right| \\
z_\ell   & = & \left| \Prefix \cap (\Upsilon^\ell)_1 \right| \; ,
\end{eqnarray*}
where $y_0 = 1$ and $z_0 = 0$ (corresponding to the empty word,
which is even). From~(\ref{eq:prefix}) we then get:
\begin{eqnarray*}
n_0 y_{\ell-1} + n_1 z_{\ell-1} & = & \eta_\ell + y_\ell \\
n_1 y_{\ell-1} + n_0 z_{\ell-1} & = & \omega_\ell + z_\ell \; ,
\end{eqnarray*}
thereby completing the proof of the ``only if'' part.

Next, we turn to proving the ``if'' part by induction
on the value of $\length = \length(\bldeta,\bldomega)$.
We assume that~(\ref{eq:prefixdist}) holds for some
pair $(\bldy,\bldz)$ with finite support,
and we let~$\length^*$ be the largest index in the union of
the supports of~$\bldy$ and~$\bldz$.
It follows from~(\ref{eq:prefixdist}) that $\length = \length^* + 1$,
i.e., $y_\ell = z_\ell = 0$ for $\ell \ge \length$.
For the induction base $\length = 1$ we have $\length^* = 0$
and, so, $\eta_1 = n_0$ and $\omega_1 = n_1$, corresponding to
$\List = \Upsilon$.

Suppose now that $\length > 1$ and define
pairs $(\bldeta',\bldomega')$ and $(\bldy',\bldz')$ as follows:
\[
\eta'_\ell =
\left\{
\begin{array}{ll}
\eta_\ell                          & \textrm{if $\ell < \length-1$} \\
\eta_{\length-1} + y_{\length-1}   & \textrm{if $\ell = \length-1$} \\
0                                  & \textrm{if $\ell > \length-1$}
\end{array}
\right.
\; ,
\]
\[
\omega'_\ell =
\left\{
\begin{array}{ll}
\omega_\ell                        & \textrm{if $\ell < \length-1$} \\
\omega_{\length-1} + z_{\length-1} & \textrm{if $\ell = \length-1$} \\
0                                  & \textrm{if $\ell > \length-1$}
\end{array}
\right.
\; ,
\]
and
\[
y'_\ell =
\left\{
\begin{array}{ll}
y_\ell    & \textrm{if $\ell \ne \length-1$} \\
0         & \textrm{if $\ell = \length-1$}
\end{array}
\right.
\; ,
\quad
z'_\ell =
\left\{
\begin{array}{ll}
z_\ell    & \textrm{if $\ell \ne \length-1$} \\
0         & \textrm{if $\ell = \length-1$}
\end{array}
\right.
\; .
\]
It can be easily verified that those pairs
satisfy~(\ref{eq:prefixdist}), namely, for every $\ell \ge 1$:
\[
\renewcommand{\arraystretch}{1.2}
\begin{array}{rcl}
\displaystyle
\eta'_\ell   & = & n_0 y'_{\ell-1} + n_1 z'_{\ell-1} - y'_\ell \\
\displaystyle
\omega'_\ell & = & n_1 y'_{\ell-1} + n_0 z'_{\ell-1} - z'_\ell
\; .
\end{array}
\]
Moreover,
$\length(\bldeta',\bldomega') < \length = \length(\bldeta,\bldomega)$.
Hence, by the induction hypothesis, there exists
an exhaustive prefix-free list~$\List'$
whose length distribution is $(\bldeta',\bldomega')$.
We construct from~$\List'$ a new list~$\List$ as follows.
We select a subset
$\Prefix_{\length-1} \subseteq \List' \cap \Upsilon^{\length-1}$
consisting of $y_{\length-1}$ arbitrary words out of
the $\eta'_{\length-1} = \eta_{\length-1} + y_{\length-1}$
words in $\List' \cap (\Upsilon^{\length-1})_0$,
and $z_{\length-1}$ additional words out of
the $\omega'_{\length-1} = \omega_{\length-1} + z_{\length-1}$
words in $\List' \cap (\Upsilon^{\length-1})_1$.
We then replace each word $\bldw \in \Prefix_{\length-1}$
by the $n_0 + n_1$  words $\bldw s$, where $s \in \Upsilon$,
i.e.,
\[
\List =
\left( \List' \setminus \Prefix_{\length-1} \right)
\cup
\left\{
\bldw s \;:\; s \in \Upsilon, \bldw \in \Prefix_{\length-1} \right\}
\; .
\]
The list~$\List$ is both exhaustive and prefix-free, and
it satisfies:
\[
\left| \List \cap (\Upsilon^\ell)_0 \right| =
\left\{
\begin{array}{ll}
\eta'_\ell                            &\textrm{if $\ell < \length-1$} \\
\eta'_{\length-1} - y_{\length-1}     &\textrm{if $\ell = \length-1$} \\
n_0 y_{\length-1} + n_1 z_{\length-1} &\textrm{if $\ell = \length$} \\
0                                     &\textrm{if $\ell > \length$}
\end{array}
\right.
\; ,
\]
namely, 
$\left| \List \cap (\Upsilon^\ell)_0 \right| = \eta_\ell$
for all $\ell \ge 1$.
In a similar way we also have
$\left| \List \cap (\Upsilon^\ell)_1 \right| = \omega_\ell$,
thereby completing the proof of the ``if'' part.
\end{proof}

\begin{remark}
\label{rem:prefixdist}
  From~(\ref{eq:prefixdist}) we get
\begin{eqnarray*}
\sum_{\ell \ge 1} \frac{\eta_\ell + \omega_\ell}{(n_0 + n_1)^\ell}
& = &
\sum_{\ell \ge 1}
\frac{(n_0{+}n_1)(y_{\ell-1}{+}z_{\ell-1})
- (y_\ell{+}z_\ell)}{(n_0 + n_1)^\ell} \\
& = &
\sum_{\ell \ge 1}
\left(
\frac{y_{\ell-1} + z_{\ell-1}}{(n_0 + n_1)^{\ell-1}}
- \frac{y_\ell + z_\ell}{(n_0 + n_1)^\ell} \right) \\
& = &
1 \; ,
\end{eqnarray*}
consistently with~(\ref{eq:Kraft}) (or with~(\ref{eq:Kplus})).\qed
\end{remark}

\begin{lemma}
\label{lem:backwardrecurrence}
Given positive integers~$n_0$ and~$n_1$,
let $(\bldeta,\bldomega)$ be a pair of
nonnegative integer sequences, each with finite support.
Then~(\ref{eq:prefixdist}) is satisfied
by a unique pair of real sequences
$\left(\bldy{=}(y_\ell)_{\ell\ge 0},\bldz{=}(z_\ell)_{\ell\ge 0}\right)$
of finite support,
and the values~$y_\ell$ and~$z_\ell$ are determined
for every $\ell \ge 1$ by
(the unique solution for $(y_\ell,z_\ell)$ of)
the following two equations:
\begin{equation}
\label{eq:backwardrecurrence+}
y_\ell + z_\ell =
\sum_{i \ge 1} \frac{\eta_{\ell+i} + \omega_{\ell+i}}{(n_0 + n_1)^i}
\end{equation}
and
\begin{equation}
\label{eq:backwardrecurrence-}
y_\ell - z_\ell =
\left\{
\renewcommand{\arraystretch}{1.5}
\begin{array}{ll}
\omega_\ell - \eta_\ell 
& \textrm{if $n_0 = n_1$} \\
\displaystyle
\sum_{i \ge 1} \frac{\eta_{\ell+i} - \omega_{\ell+i}}{(n_0 - n_1)^i}
& \textrm{if $n_0 \ne n_1$} \\
\end{array}
\right.
\; .
\end{equation}
\end{lemma}

\begin{proof}
Replacing~$\ell$ by~$\ell+1$ in~(\ref{eq:prefixdist})
and then adding (\respectively, subtracting) the two equations
in~(\ref{eq:prefixdist}), we obtain:
\[
\eta_{\ell+1} \pm \omega_{\ell+1} =
(n_0 \pm n_1) (y_\ell \pm z_\ell) - (y_{\ell+1} \pm z_{\ell+1}) \; .
\]
This, in turn, yields the following backward recurrence
for the values of $y_\ell \pm z_\ell$
(where we assume that $n_0 \ne n_1$ in the recurrence for
$y_\ell - z_\ell$):
\[
y_\ell \pm z_\ell =
\frac{\eta_{\ell+1} \pm \omega_{\ell+1}}{n_0 \pm n_1}
+
\frac{y_{\ell+1} \pm z_{\ell+1}}{n_0 \pm n_1} \; .
\]
Finally, we get~(\ref{eq:backwardrecurrence+})
and~(\ref{eq:backwardrecurrence-})
by repeated substitution, assuming the initial condition
$y_\ell = z_\ell = 0$ for any sufficiently large
$\ell \ge \length(\bldeta,\bldomega)$.
When $n_0 = n_1$, we get~(\ref{eq:backwardrecurrence-})
directly simply by
subtracting the two equations in~(\ref{eq:prefixdist}).
\end{proof}

\begin{corollary}
\label{cor:backwardrecurrence}
Using the notation of Lemma~\ref{lem:backwardrecurrence},
the pair $(\bldy,\bldz)$ satisfies~(\ref{eq:prefixdist})
for $\ell = 1$ with $(y_0,z_0) = (1,0)$, \ifandonlyif\
\begin{equation}
\label{eq:backwardrecurrence1}
\sum_{\ell \ge 1} \frac{\eta_\ell + \omega_\ell}{(n_0 + n_1)^\ell} = 1
\phantom{\; .}
\end{equation}
and (when $n_0 \ne n_1$)
\begin{equation}
\label{eq:backwardrecurrence2}
\sum_{\ell \ge 1} \frac{\eta_\ell - \omega_\ell}{(n_0 - n_1)^\ell} = 1
\; .
\end{equation}
\end{corollary}

\begin{proof}
The conditions~(\ref{eq:backwardrecurrence1})--%
                                          (\ref{eq:backwardrecurrence2})
are equivalent to requiring
that~(\ref{eq:backwardrecurrence+})--(\ref{eq:backwardrecurrence-})
be consistent with the initial condition $(y_0,z_0) = (1,0)$
for $\ell = 0$.
\end{proof}

\begin{remark}
\label{rem:backwardrecurrence}
The conditions on $(\bldeta,\bldomega)$
in Lemma~\ref{lem:backwardrecurrence},
combined
with~(\ref{eq:backwardrecurrence1})--(\ref{eq:backwardrecurrence2}),
guarantee that the solutions $(y_\ell,z_\ell)$
of~(\ref{eq:backwardrecurrence+})--(\ref{eq:backwardrecurrence-})
are integer pairs for every $\ell \ge 1$;
this can be seen if---instead of
using~(\ref{eq:backwardrecurrence+})--(\ref{eq:backwardrecurrence-})---%
we compute $(y_\ell,z_\ell)$ iteratively
for $\ell = 1, 2, 3, \cdots$, using the following recurrences
(which are implied by~(\ref{eq:prefixdist})),
\[
\renewcommand{\arraystretch}{1.2}
\begin{array}{rcl}
\displaystyle
y_\ell   & = & n_0 y_{\ell-1} + n_1 z_{\ell-1} - \eta_\ell \\
\displaystyle
z_\ell & = & n_1 y_{\ell-1} + n_0 z_{\ell-1} - \omega_\ell
\; ,
\end{array}
\]
along with the initial condition $(y_0,z_0) = (1,0)$.\qed
\end{remark}

\begin{proof}[Proof of Theorem~\protect\ref{thm:lengthdist}]
By Lemma~\ref{lem:backwardrecurrence},
Corollary~\ref{cor:backwardrecurrence},
and Remark~\ref{rem:backwardrecurrence},
conditions~(i) and~(ii) 
in Lemma~\ref{lem:lengthdist} are necessary and sufficient
for having a pair of integer sequences
$(\bldy,\bldz)$ with $(y_0,z_0) = (1,0)$
that satisfies~(\ref{eq:prefixdist}).
By Lemma~\ref{lem:prefixdist}, it remains to show
that condition~(iii) in Lemma~\ref{lem:lengthdist}
is necessary and sufficient for
these sequences to be also nonnegative.
Indeed, $y_\ell$ and~$z_\ell$ are nonnegative \ifandonlyif\
\[
y_\ell + z_\ell \ge \left| y_\ell - z_\ell \right| \; ,
\]
which,
by~(\ref{eq:backwardrecurrence+})--(\ref{eq:backwardrecurrence-}),
is equivalent to~(\ref{eq:lengthdist}).
\end{proof}

As we pointed out in Remark~\ref{rem:lengthdist},
the equality $\Kraft^+ = \Kraft^- = 0$ is equivalent to
condition~(a) in Theorem~\protect\ref{thm:lengthdist}
combined with the requirement that
the inequality~(\ref{eq:conditionb})
holds for all $\ell \ge \length = \length(\bldeta,\bldmu)$.
One may wonder if the remaining $\length-1$ inequalities in
condition~(b) are independent in the sense that,
conditioning on $\Kraft^+ = \Kraft^- = 0$,
no subset of them implies the rest.
In Appendix~\ref{sec:lengthdist}, we show that this in fact holds,
with the exception of the case $n_0 = n_1 = 1$.

\section{Parity-preserving variable-length encoders}
\label{sec:PPVLE}

In this section, we provide a formal definition of
a parity-preserving variable-length encoder.
We then state a necessary and sufficient condition
for having a parity-preserving VLE which is deterministic.
\fi

\subsection{Definition of parity-preserving variable-length encoders}
\label{sec:definitionPPVLE}

Let~$S$ be a constraint over an alphabet~$\Sigma$
and assume a partition $\{ \Sigma_0, \Sigma_1 \}$ of~$\Sigma$.
Also, let $\encoder = (V,E,L)$ be a VLG,
and for every $u \in V$ and $\ell \ge 1$,
denote by $\eta_\ell(u)$ (\respectively, $\omega_\ell(u)$)
the number of edges of length~$\ell$ outgoing from~$u$ in~$\encoder$
that have even (\respectively, odd)
labels (when the labels are regarded as words over~$\Sigma$).
Writing
\[
\bldeta(u) = \left( \eta_\ell(u) \right)_{\ell \ge 1}
\quad \textrm{and} \quad
\bldomega(u) = \left( \omega_\ell(u) \right)_{\ell \ge 1}
\; ,
\]
the pair $\left( \bldeta(u) , \bldomega(u) \right)$
thus stands for the length distribution of the set of labels of
the outgoing edges from~$u$ in~$H$.

Fix now~$n_0$ and~$n_1$ to be positive integers,
and for every $u \in V$ define
\[
\Kraft_\ell^\pm(u)
= \Kraft_\ell(\bldeta(u) \pm \bldomega(u), n_0 \pm n_1)
\]
and
\[
\Kraft^\pm(u)
= \Kraft_\length(\bldeta(u) \pm \bldomega(u), n_0 \pm n_1) \; ,
\]
where
$\length=\length(u) = \length(\bldeta(u),\bldomega(u))$.
We say that~$\encoder$ is
a \emph{(parity-preserving) $(S,n_0,n_1)$-VLE}
if for every $u \in V$ it satisfies the three conditions~(E1)--(E3)
in \Section~\ref{sec:vle},
as well as the following fourth condition:
\begin{list}{}{\settowidth{\labelwidth}{\textup{(E4)}}%
               \settowidth{\leftmargin}{\textup{(E4...)}}}
\item[(E4)]
\ifPAGELIMIT
    $\Kraft^+_\ell(u) \ge \left| \Kraft_\ell^-(u) \right|$
    for every $\ell \ge 1$.
    \end{list}
\else
for every $\ell \ge 1$: 
\[
\Kraft^+_\ell(u) \ge \left| \Kraft_\ell^-(u) \right| \; .
\]
\end{list}
(We note that condition~(E3) can be rewritten as:
\begin{list}{}{\settowidth{\labelwidth}{\textup{(E3)}}%
               \settowidth{\leftmargin}{\textup{(E3...)}}}
\item[(E3)]
$\Kraft^+(u) = 0$
\end{list}
and, so, by~(E4) we also have $\Kraft^-(u) = 0$.)
\fi

Now, let~$\Upsilon$ be a base tag alphabet of size $n_0 + n_1$
that has a partition $\{ \Upsilon_0, \Upsilon_1 \}$
with $|\Upsilon_0| = n_0$ and $|\Upsilon_1| = n_1$.
A (parity-preserving) tagging of an $(S,n_0,n_1)$-VLE is
an assignment of input tags to the edges of~$\encoder$
such that conditions~(T1)--(T2) in \Section~\ref{sec:vle} hold,
and, in addition:
\begin{list}{}{\settowidth{\labelwidth}{\textup{(T3)}}%
               \settowidth{\leftmargin}{\textup{(T3...)}}}
\item[(T3)]
at each edge, the parity of the input tag (as a word over~$\Upsilon$)
is the same as the parity of the label (as a word over~$\Sigma$).
\end{list}

It follows from Theorem~\ref{thm:lengthdist}
and conditions (E3)--(E4) that every $(S,n_0,n_1)$-VLE can be
tagged consistently with~(T3).

\subsection{Deterministic parity-preserving variable-length encoders}
\label{sec:deterministicPPVLE}

The main result of this section
is Theorem~\ref{thm:ppprincipalstates} below, which
is the parity-preserving counterpart of
Theorem~\ref{thm:principalstates}:
it presents a necessary and sufficient condition
for having a deterministic parity-preserving VLE.

Let~$\Sigma$ be an alphabet which is partitioned into
$\{ \Sigma_0, \Sigma_1 \}$
and let $H = (V,E,L)$ be a VLG whose labels are over~$\Sigma$.
Fix some nonempty subset $V' \subseteq V$
and positive integers~$n_0$ and~$n_1$,
and for every $u \in V'$ and $\ell \ge 1$,
let $(\bldeta(u|V'),\bldomega(u|V'))$
be the length distribution of the set of labels of
the outgoing edges from~$u$
in the subgraph $H' = (V',E',L')$ of~$H$ that is induced by~$V'$.
Also,
\ifPAGELIMIT
    (re-)define
\else
for the purposes of this section, redefine
\fi
\[
\Kraft_\ell^\pm(u)
= \Kraft_\ell \left( \bldeta(u|V')
\pm \bldomega(u|V'), n_0 \pm n_1 \right)
\]
and
\[
\Kraft^\pm(u)
= \Kraft_\length
\left( \bldeta(u|V') \pm \bldomega(u|V'), n_0 \pm n_1 \right) \; ,
\]
where
\ifPAGELIMIT
    $\length=\length(u)=\length(\bldeta(u|V'),\bldomega(u|V'))$.
\else
$\length=\length(u)=\length(\bldeta(u|V'),\bldomega(u|V'))$.\footnote{%
That is, $\Kraft_\ell^\pm(u)$, $\Kraft^\pm(u)$, and $\length(u)$
are redefined here for the subgraph~$H'$ of~$H$ that is induced
by the subset $V' \subseteq V$. For simplicity of notation,
we have elected to make the dependence on~$V'$ only implicit,
as~$V'$ will be understood from the context.}
\fi
We say that~$V'$ is
a set of \emph{(parity-preserving) principal states in~$H$}
\withrespectto\ $(n_0,n_1)$ if for every $u \in V'$:
\ifPAGELIMIT
    \[
    \Kraft^+(u) \le - \left| \Kraft^-(u) \right|
    \]
    and
    \[
    \Kraft^+_\ell(u) \ge \left| \Kraft^-_\ell(u) \right| \; ,
    \quad \ell = 1, 2, \ldots, \length(u)-1 \; .
    \]
\else
\begin{equation}
\label{eq:C1}
\Kraft^+(u) \le - \left| \Kraft^-(u) \right|
\end{equation}
and
\begin{equation}
\label{eq:C2}
\Kraft^+_\ell(u) \ge \left| \Kraft^-_\ell(u) \right| \; ,
\quad \ell = 1, 2, \ldots, \length(u)-1 \; .
\end{equation}
Clearly, $V'$ is a set of principal states in a VLG~$H$
(\withrespectto\ $(n_0,n_1)$), \ifandonlyif\ it is also so
in the subgraph~$H'$ of~$H$ that is induced by~$V'$.

For the special case where~$H$ is
a deterministic $(S,n_0,n_1)$-VLE,
conditions (E3)--(E4) imply that all the states of~$H$
form a set of principal states \withrespectto\ $(n_0,n_1)$,
with~(\ref{eq:C1}) replaced by the stronger condition
\begin{equation}
\label{eq:C0}
\Kraft^+(u) = \Kraft^-(u) = 0 \; .
\end{equation}
\fi

\begin{theorem}
\label{thm:ppprincipalstates}
Let~$S$ be an irreducible constraint over
an alphabet~$\Sigma$, assume a partition
$\{ \Sigma_0,\Sigma_1 \}$ of~$\Sigma$,
and let~$n_0$, $n_1$, and~$r$ be positive integers.
There exists a deterministic $(S,n_0,n_1)$-VLE
whose edges all have length at most~$r$,
\ifandonlyif~$S$ is presented
by an irreducible deterministic VLG $H = (V,E,L)$
whose edges all have length at most~$r$, and~$V$
contains a subset of principal states \withrespectto\ $(n_0,n_1)$.
\end{theorem}

\ifPAGELIMIT
    The proof of the theorem builds upon
    an (alternate) proof of Theorem~\ref{thm:principalstates};
    both proofs can be found in~\cite{RS2}.
\else
\begin{proof}
The proof of the ``only if'' part builds upon the respective part
in the proof of Theorem~\ref{thm:principalstates}.
Specifically, given a deterministic $(S,n_0,n_1)$-VLE $\encoder$,
we define the set~$V'$ as in that proof and construct
the VLE $H^* = (V',E^*,L^*)$.
For every $u \in V'$, the length distribution of the set of labels
of the outgoing edges from~$u$ in~$H^*$ is the same as
the respective set for~$Z(u)$ in~$\encoder$.
Hence, by conditions (E3)--(E4) it follows that
$H^*$ satisfies conditions~(\ref{eq:C2}) and~(\ref{eq:C0}).
Then, when we form~$H$ from~$H^*$, the change made
at states $u \in V'$ is limited to adding outgoing edges of
length~$\length(u)$.
Clearly, such a change has no effect on
the terms appearing in~(\ref{eq:C2}).
As for the terms in~(\ref{eq:C2}), 
let~$y^+$ (\respectively, $y^-$) be the number of
even-labeled (\respectively, odd-labeled)
outgoing edges that were added to state~$u$
(all of which of length $\length(u)$).
By~(\ref{eq:C0}) (when stated for $H^*$) we get that, in~$H$,
\[
\Kraft^+(u) = -y^+ - y^-
\quad \textrm{and} \quad
\Kraft^-(u) = -y^+ + y^- \; ,
\]
thereby implying~(\ref{eq:C1})
(when stated for~$H$, yet still \withrespectto\ the subset~$V'$).

Turning to the ``if'' part of the proof of
Theorem~\ref{thm:principalstates},
we need to show that we can remove edges from
the subgraph~$H'$ of~$H$ that is induced by
the set of principal states~$V'$
so that the resulting subgraph~$\encoder$
satisfies~(\ref{eq:C2}) and~(\ref{eq:C0}).
Fix some state $u \in V'$ in~$H'$, and
suppose that we remove~$y^+$ (\respectively, $y^-$)
even-labeled (\respectively, odd-labeled)
outgoing edges from state~$u$, all of length $\length = \length(u)$.
Similarly to what we had in the ``only if'' proof,
such removal does not affect the terms in~(\ref{eq:C2}),
yet it changes the values of
$\Kraft^+ = \Kraft^+(u)$
and $\Kraft^- = \Kraft^-(u)$ into
$\Kraft^+ + y^+ + y^-$ and
$\Kraft^- + y^+ - y^-$, respectively;
so, in order to satisfy~(\ref{eq:C0}),
we require that~$y^+$ and~$y^-$ be such that
\[
\Kraft^\pm + y^+ \pm y^- = 0 \; ,
\]
namely,
\begin{equation}
\label{eq:y}
y^\pm = -\frac{1}{2}
( \Kraft^+ \pm \Kraft^- ) \; .
\end{equation}
Noting that~$\Kraft^+$ and~$\Kraft^-$
have the same parity, it follows
that~$y^\pm$ satisfying~(\ref{eq:y}) are integers.
Moreover, by condition~(\ref{eq:C1}) they are also nonnegative.

To complete the proof, it remains to show that there indeed
exist $y^\pm$ edges that can be removed from~$H'$ at state~$u$,
namely, that $y^+ \le \eta_\length$ and $y^- \le \omega_\length$.
Observing that
\[
\Kraft^\pm = (n_0 \pm n_1) \Kraft_{\length-1}^\pm
- (\eta_\length \pm \omega_\length) \; ,
\]
we have:
\begin{eqnarray*}
y^\pm &
\stackrel{(\ref{eq:y})}{=} & -\frac{1}{2} (\Kraft^+ \pm \Kraft^-) \\
& = &
-\frac{1}{2}
\Bigl(
(n_0 + n_1) \Kraft_{\length-1}^+
- (\eta_\length + \omega_\length)
\Bigr. \\
&&
\quad {}
\pm
\Bigl.
(n_0 - n_1) \Kraft_{\length-1}^-
\mp (\eta_\length - \omega_\length)
\Bigr)
\; .
\end{eqnarray*}
Hence,
\begin{eqnarray*}
y^+
& = &
\eta_\length
- \frac{1}{2}
\left(
(n_0 + n_1)  \Kraft_{\length-1}^+
+ (n_0 - n_1) \Kraft_{\length-1}^- \right) \\
& \stackrel{(\ref{eq:C2})}{\le} &
\eta_\length
- \frac{1}{2}
\left(
(n_0 + n_1)  \left| \Kraft_{\length-1}^- \right|
+ (n_0 - n_1) \Kraft_{\length-1}^- \right) \\
& \le & \eta_\length \; ,
\end{eqnarray*}
with the first (\respectively, second) inequality holding with equality
\ifandonlyif\ %
$\Kraft^+_{\length-1} = \left| \Kraft^-_{\length-1} \right|$
(\respectively, $\Kraft^-_{\length-1} = 0$);
namely, $y^+ = \eta_\length$ \ifandonlyif\ %
$\Kraft^+_{\length-1} = \Kraft^-_{\length-1} = 0$.
Similarly,
\begin{eqnarray*}
y^- 
& = &
\omega_\length
- \frac{1}{2}
\left(
(n_0 + n_1)
\Kraft_{\length-1}^+
- (n_0 - n_1) \Kraft_{\length-1}^- 
\right) \\
& \stackrel{(\ref{eq:C2})}{\le} &
\omega_\length
- \frac{1}{2}
\left(
(n_0 + n_1)
\left| \Kraft_{\length-1}^- \right|
- (n_0 - n_1) \Kraft_{\length-1}^- 
\right) \\
& \le & \omega_\length \; ,
\end{eqnarray*}
again, with $y^- = \omega_\length$ \ifandonlyif\ %
$\Kraft^+_{\length-1} = \Kraft^-_{\length-1} = 0$.

We conclude that conditions~(\ref{eq:C1})--(\ref{eq:C2})
guarantee that we can always remove edges from state~$u \in V'$
in~$H'$ so that the resulting graph satisfies
((\ref{eq:C2}) and)~(\ref{eq:C0});
note that this applies also to the case
$y^+ = \eta_\length$ and $y^- = \omega_\length$,
where the edge removal reduces the value of $\length(u)$,
yet~(\ref{eq:C0}) will still hold since
$\Kraft^+_{\length-1} = \Kraft^-_{\length-1} = 0$
(see Remark~\ref{rem:lengthdist}).
\end{proof}
\fi

\begin{example}
\label{ex:2infppprincipalstates}
Let~$S$ be the $(2,\infty)$-RLL constraint,
which is presented by the graph~$G$ in \Figure~\ref{fig:2inf}.
Recall from Example~\ref{ex:2inf}
that there is no deterministic $(S^2,2)$-encoder
in this case and, so, there is no VLG~$H$ that satisfies
the conditions of Theorem~\ref{thm:principalstates} for $r = 1$.

Turning to $r = 2$, recall from Example~\ref{ex:2infprincipalstates}
that the VLE in \Figure~\ref{fig:2infvle}
is the unique induced subgraph~$H'$ of any (reduced) VLG~$H$
that satisfies the conditions of Theorem~\ref{thm:principalstates}.
Yet, assuming the ordinary definition of parity of binary words,
the set of states $V' = \{ \gamma \}$ of~$H'$
is not a set of (parity-preserving) principal states
(in $H'$ and therefore in~$H$) \withrespectto\ $(n_0,n_1) = (1,1)$.
Hence, for $r = 2$, there is no deterministic $(S^2,1,1)$-VLE.

On the other hand, there exists a deterministic $(S^2,1,1)$-VLE
for $r = 3$, as shown in \Figure~\ref{fig:2infvlealt},
along with the tag assignment in Table~\ref{tab:2infvlealt}.
\newcommand{\vle}{
    \put(000,000){\circle{20}}

    \qbezier(-9.26,004)(-30,015)(-30,000)
    \qbezier(-9.26,-04)(-30,-15)(-30,000)
    \put(-9.26,004){\vector(2,-1){0}}

    \qbezier(9.26,004)(030,015)(030,000)
    \qbezier(9.26,-04)(030,-15)(030,000)
    \put(9.26,004){\vector(-2,-1){0}}

    \qbezier(004,9.26)(015,035)(000,035)
    \qbezier(-04,9.26)(-15,035)(000,035)
    \put(004,9.26){\vector(-1,-3){0}}

    \qbezier(007,7.14)(030,055)(000,055)
    \qbezier(-07,7.14)(-30,055)(000,055)
    \put(007,7.14){\vector(-2,-5){0}}

    \put(000,000){\makebox(0,0){$\gamma$}}
    \put(-32,000){\makebox(0,0)[r]{$00$}}
    \put(032,000){\makebox(0,0)[l]{$01{.}00$}}
    \put(000,058){\makebox(0,0)[b]{$10{.}00{.}00$}}
    \put(000,038){\makebox(0,0)[b]{$10{.}01{.}00$}}
}
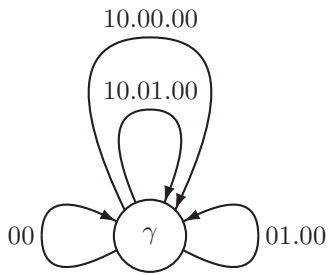
\begin{figure}[hbt]
\begin{center}
\thicklines
\setlength{\unitlength}{\figunit}
\figfont
\ifPAGELIMIT
    \begin{picture}(080,060)(-40,000)
\else
\begin{picture}(080,080)(-40,-10)
\fi
    \put(000,000){\vle}
\end{picture}
\thinlines
\setlength{\unitlength}{1pt}
\end{center}
\caption{Parity-preserving VLE for the $(2,\infty)$-RLL constraint.}
\label{fig:2infvlealt}
\end{figure}
\begin{table}[hbt]
\caption{Tag assignment for
the encoder in \Figure~\ref{fig:2infvlealt}.}
\label{tab:2infvlealt}
\ifPAGELIMIT
    \vspace{-2ex}
\fi
\[
\begin{array}{lcl}
0   & \leftrightarrow & 00           \\
10  & \leftrightarrow & 01{.}00      \\
110 & \leftrightarrow & 10{.}01{.}00 \\
111 & \leftrightarrow & 10{.}00{.}00
\end{array}
\]
\end{table}
\ifPAGELIMIT
\else
This encoder is a subgraph of the VLG~$H$
shown in \Figure~\ref{fig:2infvlg}, which is a deterministic
VLG presentation of~$S^2$ whose edges all have length at most~$3$,
and $V' = \{ \gamma \}$ is a set of principal states in~$H$
\withrespectto\ $(1,1)$ (as such, $H$ satisfies the conditions of
Theorem~\ref{thm:ppprincipalstates}).
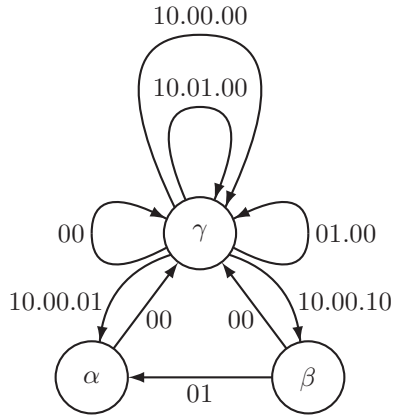
\begin{figure}[t]
\begin{center}
\thicklines
\setlength{\unitlength}{\figunit}
\figfont
\begin{picture}(120,120)(-60,-10)
    \put(000,040){\vle}
    \put(-30,000){
        \multiput(000,000)(060,000){2}{\circle{20}}
        \put(050,000){\vector(-1,0){40}}
        \put(006,008){\vector(3,4){18}}
        \put(054,008){\vector(-3,4){18}}

        \qbezier(38.0,34.0)(56.2,28.2)(57.8,9.6)
        \put(57.8,9.6){\vector(1,-6){0}}

        \qbezier(22.0,34.0)(3.8,28.2)(2.2,9.6)
        \put(2.2,9.6){\vector(-1,-6){0}}

        \put(000,000){\makebox(0,0){$\alpha$}}
        \put(060,-.5){\makebox(0,0){$\beta$}}

        \put(015,016){\makebox(0,0)[l]{$00$}}
        \put(045,016){\makebox(0,0)[r]{$00$}}
        \put(030,-02){\makebox(0,0)[t]{$01$}}
        \put(003,021){\makebox(0,0)[r]{$10{.}00{.}01$}}
        \put(057,021){\makebox(0,0)[l]{$10{.}00{.}10$}}
    }
\end{picture}
\thinlines
\setlength{\unitlength}{1pt}
\end{center}
\caption{Deterministic VLG presentation~$H$ of
the second power of the $(2,\infty)$-RLL constraint.}
\label{fig:2infvlg}
\end{figure}
\fi

Comparing to the fixed-length case,
using Theorem~\ref{thm:main}(a), one can verify
that there exists a (not necessarily deterministic)
$(S^{2p},2^{p-1},2^{p-1})$-encoder,
\ifandonlyif\ $p \ge 3$. For $p = 3$,
any vector $\bldx \in \XX(A_{(G^6)_0},4) \cap \XX(A_{(G^6)_1},4)$
satisfies $\| \bldx \|_\infty \ge 6$
(and equality is attained only by $\bldx = (2 \; 3 \; 6)^\top$).
By Corollaries 4 and~5 in~\cite{RS1} we then get that
any rate $3:6$ parity-preserving fixed-length encoder for~$S$
must have at least six states and
\ifPAGELIMIT
    decoding look-ahead of
\else
anticipation
\fi
at least~$2$
(measured in $6$-bit symbols); in contrast,
recall that when there is no requirement for parity preservation,
we have the simple encoder in \Figure~\ref{fig:2infencoder}.
Using Theorem~\ref{thm:main}(b), one can determine that there exists
a rate $p:2p$ parity-preserving fixed-length encoder for~$S$
which is deterministic, (if and) only if $p \ge 8$.\qed
\end{example}

\begin{remark}
\label{rem:ppprincipalstates}
Unlike Theorem~\ref{thm:principalstates},
we do not have (as of yet) 
an extension of Franaszek's algorithm from~\cite{Franaszek2}
to the parity-preserving case;
namely, a polynomial-time algorithm is yet to be found
for determining whether, for given~$S$, $\{ \Sigma_0, \Sigma_1 \}$,
$n_0$, $n_1$, and~$r$,
there is a VLG~$H$ that satisfies the conditions of
Theorem~\ref{thm:ppprincipalstates}.
(The problem, however, is still decidable, since there are only
finitely many reduced VLGs $H$ with edge lengths at most~$r$
such that $S(H) = S$.)\qed
\end{remark}

\ifPAGELIMIT
\else
\subsection{Discussion}
\label{sec:discussion}

In Appendix~\ref{sec:nonexistence}, we show that
for the constraint~$S$ of Example~\ref{ex:twostates}
and for the partition~(\ref{eq:partitioningalt}),
there is no deterministic $(S^t,n_0,n_1)$-VLE,
for any positive integers~$t$, $n_0$, and~$n_1$ such that
$\log_2 (n_0 + n_1) = t = \capacity(S^t)$.
In contrast, given any constraint~$S = S(G)$
and positive integers~$n_0$ and~$n_1$
that satisfy the strict inequality $\log_2 (n_0 + n_1) < \capacity(S)$,
it follows from (the proof of) Theorem~2 in~\cite{RS1}
that, under mild conditions on the presentation~$G$ of~$S$,
there exist deterministic (fixed-length)
$(S^r,n^{(r)},n^{(r)})$-encoders $H_r$, $r = 1, 2, \ldots$,
where
$(\log_2 n^{(r)})/r \rightarrow \capacity(S) \; (> \log_2 (n_0 + n_1))$
when $r \rightarrow \infty$.
Thus, for sufficiently large~$r$, each encoder $H_r$,
when regarded as a VLG with all the edges having length~$r$,
contains a deterministic $(S,n_0,n_1)$-VLE as a subgraph.

As a sanity check, we next show that the states of $H_r$
form a set of principal states \withrespectto\ $(n_0,n_1)$. From
$(\log_2 n^{(r)})/r > \log_2 (n_0 + n_1)$
(for sufficiently large~$r$) it follows that
\begin{equation}
\label{eq:nr}
(n_0 + n_1)^r + \left| n_0 - n_1 \right|^r \le 2 n^{(r)} \; .
\end{equation}
Now, the VLG $H_r$ (whose edges all have length~$r$)
satisfies~(\ref{eq:C2}) vacuously
(with $V'$ taken as the whole set of states of $H_r$),
and it also satisfies~(\ref{eq:C1}) since
\begin{eqnarray*}
\Kraft^+(u) & = &
(n_0 + n_1)^r - \eta_r(u) - \omega_r(u) \\
& = &
(n_0 + n_1)^r - 2 n^{(r)} \\
& \stackrel{(\ref{eq:nr})}{\le} &
- \left| n_0 - n_1 \right|^r \\
& = &
- \left| (n_0 - n_1)^r - \eta_r(u) + \omega_r(u) \right| \\
& = & 
- \left| \Kraft^-(u) \right| \; .
\end{eqnarray*}
We conclude that the states of $H_r$ form a principal set of states
and, so, by Theorem~\ref{thm:ppprincipalstates}
there exists a deterministic $(S(H_r),n_0,n_1)$-VLE
(and, as such, it is also an $(S,n_0,n_1)$-VLE).

When $G = (V,E,L)$ is an ordinary graph
(whose edges all have length~$1$), condition~(\ref{eq:C1}) becomes,
for every $u \in V'$:
\begin{eqnarray*}
\lefteqn{
(n_0 + n_1)
- \sum_{v \in V'} \left(A_{G_0} + A_{G_1} \right)_{u,v}
} \makebox[10ex]{} \\
& \le &
\Bigl|
(n_0 - n_1)
- \sum_{v \in V'} \left(A_{G_0} - A_{G_1} \right)_{u,v}
\Bigr| \; .
\end{eqnarray*}
This inequality can be rewritten as
\[
\sum_{v \in V'} \left(A_{G_0} \right)_{u,v} \ge n_0
\quad \textrm{and} \quad
\sum_{v \in V'} \left(A_{G_1} \right)_{u,v} \ge n_1 \; ,
\]
and also as
\[
A_{G_0} \bldx \ge n_0 \bldx
\quad \textrm{and} \quad
A_{G_1} \bldx \ge n_1 \bldx \; ,
\]
where~$\bldx$ is the $0\textrm{--}1$ characteristic vector of
the subset~$V'$ within~$V$.
Condition~(\ref{eq:C2}) becomes vacuous for ordinary graphs.
It thus follows that a nonempty subset $V' \subseteq V$ is
a set of principal states in~$G$ \withrespectto\ $(n_0,n_1)$,
\ifandonlyif\ its characteristic vector belongs to
$\XX(A_{G_0},n_0) \cap \XX(A_{G_1},n_1)$.
For~$G$ which is also deterministic,
this coincides with Theorem~\ref{thm:main}(b).

\begin{remark}
\label{rem:SBD}
When applying
Theorems~\ref{thm:principalstates} and~\ref{thm:ppprincipalstates}
to a finite-memory\footnote{%
A constraint~$S$ has finite memory if it can be defined
through a finite list of forbidden words,
i.e., $\bldw \in S$ if and only if~$\bldw$
does not contain any word in that list as
a sub-word~\cite[\S 2.3]{MRS}.}
constraint~$S$ and $r = 1$,
the respective (fixed-length) deterministic $(S,n_0,n_1)$-encoder can be
guaranteed to be also sliding-block decodable.
On the other hand,
when $r > 1$, edges in the encoder may have different lengths
and, so, the output sequence consists of words (labels) of
varying lengths over the alphabet~$\Sigma$ of~$S$.
State-independent decoding, however, should not assume the position of
any given output symbol (of~$\Sigma$) within the label (word)
that it belongs to. This, in turn, imposes conditions beyond
the Kraft conditions~(\ref{eq:C2})--(\ref{eq:C0}) on
the lengths of the outgoing edges from each state in the encoder.
When encoders do not have to be parity-preserving,
such (sufficient) conditions were provided
in~\cite{Beal1} and~\cite{Beal2}.
Respective conditions are yet to be found for
the parity-preserving case.\qed
\end{remark}

In this paper, we focused mainly
on parity-preserving VLEs which are deterministic.
The study of the non-deterministic case is an open topic for
future work. In particular, we can pose the following question:
under what conditions can capacity be achieved (with equality)
by parity-preserving VLEs? 
Recall that for the constraint~$S$ of Example~\ref{ex:twostates}
and for the partition~(\ref{eq:partitioningalt}),
capacity cannot be achieved when the encoder is deterministic
(as we show in Appendix~\ref{sec:nonexistence}),
nor when it is of fixed length (as we showed in~\cite{RS1}).


\ifIEEE
   \appendices
\else
   \section*{$\,$\hfill Appendices\hfill$\,$}
   \appendix
\fi

\section{Nonexistence result for Example~\protect\ref{ex:vlealt}}
\label{sec:nonexistence}

Let~$\Sigma$, $G$, and~$S$ be as in Example~\ref{ex:twostates},
and assume the partition~(\ref{eq:partitioningalt}) of~$\Sigma$.
We show that for this partition, there is no
deterministic parity-preserving VLE
at a coding ratio of~$1$. 
Specifically, we show that for every positive integers
$t$, $n_0$, and $n_1$ such that $n_0 + n_1 = 2^t$,
there is no deterministic parity-preserving
$(S^t,n_0,n_1)$-VLE.

Suppose to the contrary that such a VLE exists,
and let~$\encoder$ be such an encoder with
the smallest number of states.
The encoder~$\encoder$ is irreducible (or else 
its irreducible sink\footnote{%
Refer to Footnote~\ref{fn:sink} for the definition of
an irreducible sink.}
would be a smaller encoder)
and reduced (or else we could merge states with identical
follower sets~\cite[\S 2.6.2]{MRS}).
By changing the outgoing edges from each state in~$\encoder$
into a tree, we can get
an (ordinary) irreducible deterministic graph~$G'$
(with edge labels of length~$1$ over~$\Sigma^t$).
The constraint $S(\encoder) = S(G')$ has capacity~$t$, which is also
the capacity of the (irreducible) constraint~$S^t$
in which it is contained.
Hence, by~\cite[Problem~3.28]{MRS}
we have $S(\encoder) = S(G') = S^t = S(G^t)$
and, so, by Remark~\ref{rem:vlg},
for every state $u$ in~$\encoder$
there exists a state $v \in \{ \alpha, \beta \}$ in~$G$ such that
$\FF_\encoder(u) = \FF_{G^t}(v)$.
It follows that~$\encoder$ has no more than two states.

Assume first that~$\encoder$ has only one state,
in which case the edges in~$\encoder$ are (variable-length) self-loops,
corresponding to cycles in~$G^t$.
Note, however, that all the cycles in~$G$ (and, therefore, in~$G^t$)
generate even words, which means that any tagging
of the edges of~$\encoder$ forms a set~$\List$
consisting only of even words over the base tag alphabet~$\Upsilon$.
Yet, since we assume that both~$n_0$ and~$n_1$ are positive,
the alphabet~$\Upsilon$ contains at least one even symbol (say, $0$)
and one odd symbol (say, $1$).
But then, any word of
the form $100\ldots 0$ that is longer than the longest tag in~$\List$
is neither a prefix of any tag in~$\List$ (obviously),
nor has it a prefix in~$\List$; namely,~$\List$ cannot be
exhaustive.\footnote{The case $n_0 = 0$ can also be ruled out:
the first label along any path that generates the (sufficiently long)
even word $abdbd\ldots bd$ must end either with an~$a$ or
with a~$d$, and, so, that label must have odd length. 
On the other hand, any odd-length tag over an an all-odd alphabet
cannot be even.

We point out that we can rule out an encoder with all-even labels
also by using Lemma~\ref{lem:lengthdist}:
it is easy to see that when $\omega_\ell = 0$ for all~$\ell$,
conditions~(i) and~(ii') can hold simultaneously only when $n_1 = 0$.}

It remains to consider the case where~$\encoder$
has two (inequivalent) states,
which we denote by~$\alpha$ and~$\beta$
to match their respective equivalent states in~$G$.
In fact, we will rule out the existence of 
a deterministic two-state $(S^t,2^t)$-VLE, regardless of
whether it is parity-preserving.
Any deterministic $(S^t,2^t)$-VLE, in turn,
can be viewed as a deterministic $(S,2)$-VLE, by regarding each
length-$\ell$ label over~$\Sigma^t$
as a label of length $t \ell$ over~$\Sigma$.

We recall the following definition of a parametrized adjacency matrix.
Given a VLG $H = (V,E,L)$,
for any two states~$u, v \in V$
we denote by $\mu_\ell(u,v)$ the number of edges of
length~$\ell$ from~$u$ to~$v$.
For a positive real indeterminate~$\theta$,
we define the parametrized adjacency matrix of~$H$
as the $|V| \times |V|$ matrix~$A_H(\theta)$ whose entries are given by:
\[
(A_H(\theta))_{u,v}
= \sum_{\ell \ge 1} \mu_\ell(u,v) \cdot \theta^{-\ell} \; .
\]
We let $\theta_{\max}(H)$ denote the largest~$\theta$
for which $\lambda(A_H(\theta)) = 1$.
It is known that when~$H$ is lossless
(in particular, deterministic), the capacity of~$S(H)$
equals $\log_2 \theta_{\max}(H)$~\cite[Theorem~1]{Shannon}
(when all the edges have length~$1$
we have $A_H(\theta) = (1/\theta) \cdot A_H$,
in which case $\theta_{\max}(H) = \lambda(A_H)$).
It is also known that when $H$ is irreducible,
the mapping $\theta \mapsto \lambda(A_H(\theta))$
is strictly decreasing (and continuous)
over $(0,\infty)$~\cite[Proposition~3.12]{MRS}.

Turning now to the encoder~$\encoder$, 
which we view as a deterministic two-state $(S,2)$-VLE,
we have
\begin{eqnarray}
\sum_{v \in \{ \alpha, \beta \}}
(A_\encoder)_{\beta,v}(2)
& = &
\!\!\!\!\!\!\!
\sum_{v \in \{ \alpha, \beta \}} 
\sum_{\ell \ge 1} \mu_\ell(\beta,v) \cdot 2^{-\ell} \nonumber \\
& = &
\sum_{\ell \ge 1} 2^{-\ell} \sum_{v \in \{ \alpha, \beta \}}
\mu_\ell(\beta,v) \nonumber \\
\label{eq:rowsumbeta}
& = & 1 \; ,
\end{eqnarray}
where the last equality follows from condition~(E3).
Denote by~$\Labels$ the set of labels of
the outgoing edges from state~$\beta$ in~$\encoder$.
By~(\ref{eq:rowsumbeta})
we must have
\[
|\Labels| = \sum_{\ell \ge 1}
\sum_{v \in \{ \alpha, \beta \}} \mu_\ell(\beta,v) \ge 2 \; .
\]
Noting that all these labels start with the symbol~$d$,
we define~$\Labels'$ to be the set of all words obtained by
removing the leading symbol~$d$ from the words in~$\Labels$.
Since~$\encoder$ is deterministic, the set $\Labels$
is prefix-free and, therefore, so is~$\Labels'$.
In particular, $\Labels'$ does not contain the empty word
(since $|\Labels'| = |\Labels| \ge 2$).
Next, construct from~$\encoder$ an (irreducible) VLG~$\encoder'$
by replacing the outgoing edges from state~$\alpha$ with copies of
the outgoing edges from state~$\beta$, keeping the terminal states
yet removing from each label its leading symbol~$d$.
Thus, for every $u, v \in \{ \alpha, \beta \}$:
\begin{equation}
\label{eq:Aencoder'}
(A_{\encoder'}(\theta))_{u,v}
=
\left\{
\begin{array}{ll}
\theta \cdot (A_\encoder(\theta))_{\beta,v}
& \textrm{if $u = \alpha$} \\
(A_\encoder(\theta))_{\beta,v}
& \textrm{if $u = \beta$}
\end{array}
\right.
\; .
\end{equation}
Since $\Labels'$ is prefix-free the graph~$\encoder'$ is deterministic.
Moreover, it can be easily verified
that $\FF_{\encoder'}(\alpha) \subseteq \FF_G(\alpha)$ and, so,
$S(\encoder') \subseteq S$.
Hence, $\capacity(S(\encoder')) \le \capacity(S) = 1$,
which implies that
\begin{equation}
\label{eq:ximax}
\theta_{\max}(\encoder') \le 2 \; .
\end{equation}
On the other hand, from~(\ref{eq:rowsumbeta})--(\ref{eq:Aencoder'}) 
we get the following row sums in $A_{\encoder'}(2)$:
\[
\sum_{v \in \{ \alpha, \beta \}}
(A_{\encoder'}(2))_{\alpha,v} = 2
\quad \textrm{and} \quad
\!\!\!
\sum_{v \in \{ \alpha, \beta \}}
(A_{\encoder'}(2))_{\beta,v} = 1 \; .
\]
By~\cite[Proposition~3.13]{MRS} we then get
that $\lambda(A_{\encoder'}(2)) > 1$, i.e.,
$\theta_{\max}(\encoder') > 2$. Yet this contradicts~(\ref{eq:ximax}).

\section{Independence of the conditions in
Theorem~\protect\ref{thm:lengthdist}}
\label{sec:lengthdist}

Given positive integers $n_0$, $n_1$, and $r > 1$,
we show that, unless $n_0 = n_1 = 1$,
the inequalities~(\ref{eq:conditionb})
that correspond to $\ell = 1, 2, \ldots, r-1$ are independent
(in the sense defined below)
conditioned on $(\bldeta,\bldomega)$ being a nonnegative integer pair
with $\length = \length(\bldeta,\bldomega) = r$
that satisfies $\Kraft^+ = \Kraft^- = 0$. In particular,
each of these inequalities is necessary,
as it is not implied by the rest.

We introduce the following definition.
Given positive integers $n_0$, $n_1$, and $r > 1$,
a subset $\ZZ \subseteq \{ 1, 2, \ldots, r-1 \}$
is said to be \emph{admissible for $(n_0,n_1,r)$}
if there exists a nonnegative integer pair $(\bldeta,\bldomega)$
with $\length = \length(\bldeta,\bldomega) = r$
that satisfies $\Kraft^+ = \Kraft^- = 0$
yet violates~(\ref{eq:conditionb}) when (and only when) $\ell \in \ZZ$.
The inequalities~(\ref{eq:conditionb}) are then said to be
independent if every subset $\ZZ \subseteq \{ 1, 2, \ldots, r-1 \}$
is admissible for $(n_0,n_1,r)$.

We have the following lemma.

\begin{lemma}
\label{lem:admissible}
A subset $\ZZ \subseteq \{ 1, 2, \ldots, r-1 \}$
is admissible for $(n_0,n_1,r)$,
\ifandonlyif\ there exists an integer pair
$\left(\bldy{=}(y_\ell)_{\ell\ge 0},\bldz{=}(z_\ell)_{\ell\ge 0}\right)$
that satisfies the following conditions:
\begin{list}{}{\settowidth{\labelwidth}{\textup{(E2)}}%
               \settowidth{\leftmargin}{\textup{(E2...)}}}
\item[\textup{(C1)}]
$y_\ell \le n_0 y_{\ell-1} + n_1 z_{\ell-1}$ and
$z_\ell \le n_1 y_{\ell-1} + n_0 z_{\ell-1}$ 
for every $\ell \ge 1$,
\item[\textup{(C2)}]
$y_{r-1} + z_{r-1} > 0$, 
\item[\textup{(C3)}]
$y_\ell = z_\ell = 0$ when $\ell \ge r$,
\item[\textup{(C4)}]
$y_0 = 1$ and $z_0 = 0$, and---
\item[\textup{(C5)}]
$\min \{ y_\ell, z_\ell \} < 0$
when (and only when) $\ell \in \ZZ$.
\end{list}
\end{lemma}

\begin{proof}
We use~(\ref{eq:prefixdist})
and (\ref{eq:backwardrecurrence+})--(\ref{eq:backwardrecurrence-})
to define a one-to-one correspondence between
integer pairs $(\bldeta,\bldomega)$ and $(\bldy,\bldz)$,
both with finite support.
Condition~(C1) is equivalent to requiring that~$\bldeta$ and~$\bldomega$
are nonnegative, and conditions (C2)--(C3) are equivalent to
having $\length(\bldeta,\bldomega) = r$
(and, when $n_0 = n_1$, also $\eta_r = \omega_r$).
Conditioning on (C1)--(C3),
we get by Corollary~\ref{cor:backwardrecurrence}
that condition~(C4) is equivalent to conditions~(i) and~(ii')
in (the proof of) Lemma~\ref{lem:lengthdist}
being satisfied by $(\bldeta,\bldomega)$; these conditions, in turn,
are equivalent to requiring $\Kraft^+ = \Kraft^- = 0$.

Finally, conditioning on (C1)--(C4)
(and, in particular, on on~(i) and~(ii')), we get from 
(\ref{eq:backwardrecurrence+})--(\ref{eq:backwardrecurrence-})
that~(\ref{eq:lengthdist}) (and, therefore, (\ref{eq:conditionb}))
can be rewritten as
\[
y_\ell + z_\ell \ge \left| y_\ell - z_\ell \right| \; ,
\]
which, in turn, holds \ifandonlyif\ $y_\ell$ and $z_\ell$
are nonnegative. Hence, condition~(C5) is equivalent
to~(\ref{eq:conditionb}) being violated by $(\bldeta,\bldomega)$
when (and only when) $\ell \in \ZZ$.
\end{proof}

We now use Lemma~\ref{lem:admissible} to identify
the admissible subsets for any given $(n_0,n_1,r)$.
In particular, we show that when $\max \{ n_0, n_1 \} > 1$,
every subset $\ZZ \subseteq \{ 1, 2, \ldots, r-1 \}$ is admissible.
We distinguish between three cases.

\emph{Case~1: $n_0 \ge n_1$ and $n_0 > 1$.}
We show that any subset $\ZZ \subseteq \{ 1, 2, \ldots, r-1 \}$
is admissible for $(n_0,n_1,r)$, for any $r > 1$.
Indeed, given any such subset~$\ZZ$,
define the pair $(\bldy,\bldz)$ by:
\[
y_\ell = 
\left\{
\begin{array}{cl}
1   & \textrm{if $\ell = 0$}       \\
n_0 & \textrm{if $1 \le \ell < r$} \\
0   & \textrm{if $\ell \ge r$}
\end{array}
\right.
\]
and
\[
z_\ell =
\left\{
\begin{array}{cl}
-1 & \textrm{if $\ell \in \ZZ$} \\
0  & \textrm{otherwise}
\end{array}
\right.
\; .
\]
It can be readily checked that the pair $(\bldy,\bldz)$
satisfies conditions (C1)--(C5) (where, for (C1)--(C2),
we use the assumption that $n_0 > 1$).

\emph{Case~2: $n_1 \ge n_0$ and $n_1 > 1$.}
Here, too, any subset~$\ZZ \subseteq \{ 1, 2, \ldots, r-1 \}$
is admissible, for any $r > 1$:
the proof is similar to Case~2, except that
the pair $(\bldy,\bldz)$ is now defined by
\[
y_\ell =
\left\{
\begin{array}{cl}
1   & \textrm{if $\ell = 0$}                          \\
n_1 & \textrm{if $\ell$ is even and $1 \le \ell < r$} \\
-1  & \textrm{if $\ell$ is odd and $\ell \in \ZZ$}    \\
0   & \textrm{otherwise}
\end{array}
\right.
\]
and
\[
z_\ell = 
\left\{
\begin{array}{cl}
n_1 & \textrm{if $\ell$ is odd and $1 \le \ell < r$} \\
-1  & \textrm{if $\ell$ is even and $\ell \in \ZZ$}  \\
0   & \textrm{otherwise}
\end{array}
\right.
\; .
\]

\emph{Case~3: $n_0 = n_1 = 1$.}
In this case, there are subsets of $\{ 1, 2, \ldots, r-1 \}$
which are not admissible. For example, it can be verified that
the inequality~(\ref{eq:conditionb}) for $\ell = 1$ is implied
by $\Kraft^+ = \Kraft^- = 0$.

We next characterize the admissible subsets for $(1,1,r)$.
Given a subset $\ZZ \subseteq \{ 1, 2, \ldots, r-1 \}$,
define the integer sequence
$\bldxi = \bldxi(\ZZ,r) = (\xi_1 \; \xi_2 \ldots \xi_r)$
inductively as follows:
\begin{equation}
\label{eq:xi}
\xi_\ell = 
\left\{
\begin{array}{cl}
1                & \textrm{if $\ell = 1$} \\
\xi_{\ell-1} - 1 & \textrm{if $\ell-1 \in \ZZ$} \\
2\xi_{\ell-1}    & \textrm{otherwise}
\end{array}
\right.
\; .
\end{equation}

We have the following lemma.

\begin{lemma}
\label{lem:admissible=1}
A subset $\ZZ \subseteq \{ 1, 2, \ldots, r-1 \}$
is admissible for $(1,1,r)$,
\ifandonlyif\ the sequence $\bldxi$
as defined in~(\ref{eq:xi}) is all-positive.
\end{lemma}

\begin{proof}
Starting with the ``only if'' part,
suppose that there exists an integer pair $(\bldy,\bldz)$
that satisfies conditions (C1)--(C5). Condition~(C1) can be
rewritten as
\begin{equation}
\label{eq:admissible=1}
\max \{ y_\ell, z_\ell \} \le y_{\ell-1} + z_{\ell-1}
\end{equation}
which, with conditions~(C4)--(C5), implies (by induction on~$\ell$)
that $\max \{ y_\ell, z_\ell \} \le \xi_\ell$ for every
$\ell \in \{ 1, 2, \ldots, r-1 \}$.
In particular, for $\ell = r-1$ we have:
\begin{eqnarray*}
\xi_{r-1} & \ge & \max \{ y_{r-1}, z_{r-1} \} \\
& = &
y_{r-1} + z_{r-1} - \min \{ y_{r-1}, z_{r-1} \} \\
& \stackrel{\mathrm{(C2)}}{\ge} &
1 - \min \{ y_{r-1}, z_{r-1} \} \\
& \stackrel{\mathrm{(C5)}}{\ge} &
\left\{
\begin{array}{cl}
2 & \textrm{if $r-1 \in \ZZ$}\\
1 & \textrm{otherwise}
\end{array}
\right.
\; ,
\end{eqnarray*}
which, by~(\ref{eq:xi}), implies that $\xi_r > 0$.
Moreover, by~(\ref{eq:xi}), the inequality $\xi_{r-1} > 0$
is possible only if $\xi_\ell > 0$ for every $\ell < r$.

Turning to the ``if'' part,
given a sequence~$\bldxi$
as in~(\ref{eq:xi}) that is all-positive,
we define the pair $(\bldy,\bldz)$ as follows:
\[
y_\ell =
\left\{
\begin{array}{cl}
1        & \textrm{if $\ell = 0$}       \\
\xi_\ell & \textrm{if $1 \le \ell < r$} \\
0        & \textrm{otherwise}
\end{array}
\right.
\]
and
\[
z_\ell = 
\left\{
\begin{array}{cl}
-1       & \textrm{if $\ell \in \ZZ$} \\
\xi_\ell & \textrm{if $\ell \in \{ 1, 2, \ldots,r-1\} \setminus \ZZ$} \\
0        & \textrm{otherwise}
\end{array}
\right.
\; .
\]
Obviously, the pair $(\bldy,\bldz)$ satisfies conditions (C3)--(C5).
As for condition~(C2), we have
\begin{eqnarray*}
y_{r-1} + z_{r-1} & = &
\left\{
\begin{array}{cl}
\xi_{r-1} - 1 & \textrm{if $r-1 \in \ZZ$}\\
2\xi_{r-1}    & \textrm{otherwise}
\end{array}
\right.
\\
& = &
\xi_r > 0 \; .
\end{eqnarray*}
Turning finally to condition~(C1),
the inequality~(\ref{eq:admissible=1}) holds (trivially)
with equality when $\ell = 1$ or when $\ell > r$,
and is implied by condition~(C2) when $\ell = r$.
For the remaining range $\ell \in \{ 2, 3, \ldots, r-1 \}$
we also have equality in~(\ref{eq:admissible=1}), since:
\begin{eqnarray*}
\max \{ y_\ell, z_\ell \} & = & y_\ell
= \xi_\ell
\stackrel{(\ref{eq:xi})}{=}
\left\{
\begin{array}{cl}
\xi_{\ell-1} - 1 & \textrm{if $\ell-1 \in \ZZ$}\\
2\xi_{\ell-1}    & \textrm{otherwise}
\end{array}
\right.
\\
& = &
y_{\ell-1} + z_{\ell-1} \; .
\end{eqnarray*}
\end{proof}
\fi

\end{document}